\documentclass[reqno,12pt,letterpaper]{amsart}
\usepackage[proof]{newmzmacros3}
\usepackage[all,cmtip]{xy}

\newcommand{\RR}{{\mathbb R}}
\newcommand{\NN}{{\mathbb N}}

\usepackage{soul}
\usepackage{xcolor}
\usepackage{fancyvrb}
\usepackage{subcaption}

\title{}

\author{Jeffrey Galkowski}
\email{j.galkowski@ucl.ac.uk}
\address{Department of Mathematics, University College London, WC1H 0AY, UK}

\author{Maciej Zworski} 
\email{zworski@math.berkeley.edu}
\address{Department of Mathematics, University of California, Berkeley, CA 94720}

\email{hertz@math.berkeley.edu}

\address{Department of Mathematics, University of California,
Berkeley, CA 94720, USA.}

\begin{document}

\title{Classical--Quantum correspondence in Lindblad evolution}

\maketitle

\vspace{-0.2in}

\begin{center}
{\sc With an appendix by Zhen Huang and Maciej Zworski}
\end{center}

\vspace{-0.1in}

\begin{abstract}
We show that for the Lindblad evolution defined using (at most) quadratically growing 
classical Hamiltonians and (at most) linearly growing classical jump functions (quantized into 
jump operators assumed to satisfy certain ellipticity conditions and modeling interaction with a larger system), the evolution of a quantum observable remains close to the classical 
Fokker--Planck evolution in the Hilbert--Schmidt norm for times vastly exceeding the Ehrenfest time
(the limit of such agreement with no jump operators). The time scale is the same as in the recent 
papers \cite{hrr,hrrb}  by Hern\'andez--Ranard--Riedel but the statement and methods are different. 
{The appendix presents numerical experiments illustrating the classical/quantum correspondence
in Lindblad evolution and comparing it to the mathematical results.}
\end{abstract}

\section{Introduction}
\label{s:intr}

In quantum mechanics a system is often described using a {\em density matrix}, that is a
 positive operator of trace one on a Hilbert space. In this paper the Hilbert space will be
given by $ L^2 ( \mathbb R^n ) $ so that the density operator is then 
\[      A u ( x ) = \sum_{j} p_j \langle u , u_j \rangle u_j ( x ) , \ \ \ \ p_j \geq 0, \ \ \  \sum_{j} p_j = 1, \ \ \ 
\langle u_j , u_i \rangle = \delta_{ij}  . \]
If the system evolves according to the Schr\"odinger equation $ ( i h \partial_t + P ) v( t ) = 0 $, 
where $ P $ is a self-adjoint unbounded operator on $ L^2 (\mathbb R^n ) $ then 
(note the sign convention) the density matrix evolves by the Schr\"odinger propagation of $ u_j $'s. That gives the following equation:
\begin{equation}
\label{eq:Schr} \partial_t A (t) = \mathcal L_0 A (t) ,  \ \ \mathcal L_0 A := \frac i h [P , A ] , \ \ \
A ( t ) = e^{ t \mathcal L_0} A (0) = e^{ i t P/h} A (0) e^{ - i t P/h }  . 
\end{equation}
This evolution clearly preserves density matrices. 
Gorini--Kossakowski--Sudarshan \cite{gks} and Lindblad \cite{lin} generalized this by showing  (in the setting of
matrices and of bounded operators, respectively) that semigroups preserving
the trace and complete positivity are generated by operators of the form
\begin{equation}
\label{e:Lindblad}
\mathcal L A := \frac{i}{h}[P,A] +\frac{\gamma}{h}\sum_{j=1}^J  \left( L_jA L_j^*-\tfrac{1}{2}(L_j^*L_j A +A L_j^*L_j) \right), 
 \ \ \ \gamma \geq 0. 
\end{equation}
The corresponding evolution equation is called the Lindblad master equation or 
the GKLS equation and, following the long tradition which favours short northern European names, 
we refer to $ \mathcal L $ as the Lindbladian -- see \cite{chp} for a history of this discovery and pointers to the literature. The operators $ L_j $ are called jump operators and they describe a dissipative (see 
\eqref{eq:FP} below) interaction of a system evolving according to \eqref{eq:Schr} with a larger ``open"
system. (Hence the term ``jump'' as $ L_j $ describe the effect of moving to that larger system.)

\begin{figure}[h]
  \centering
  \includegraphics[width=170mm]{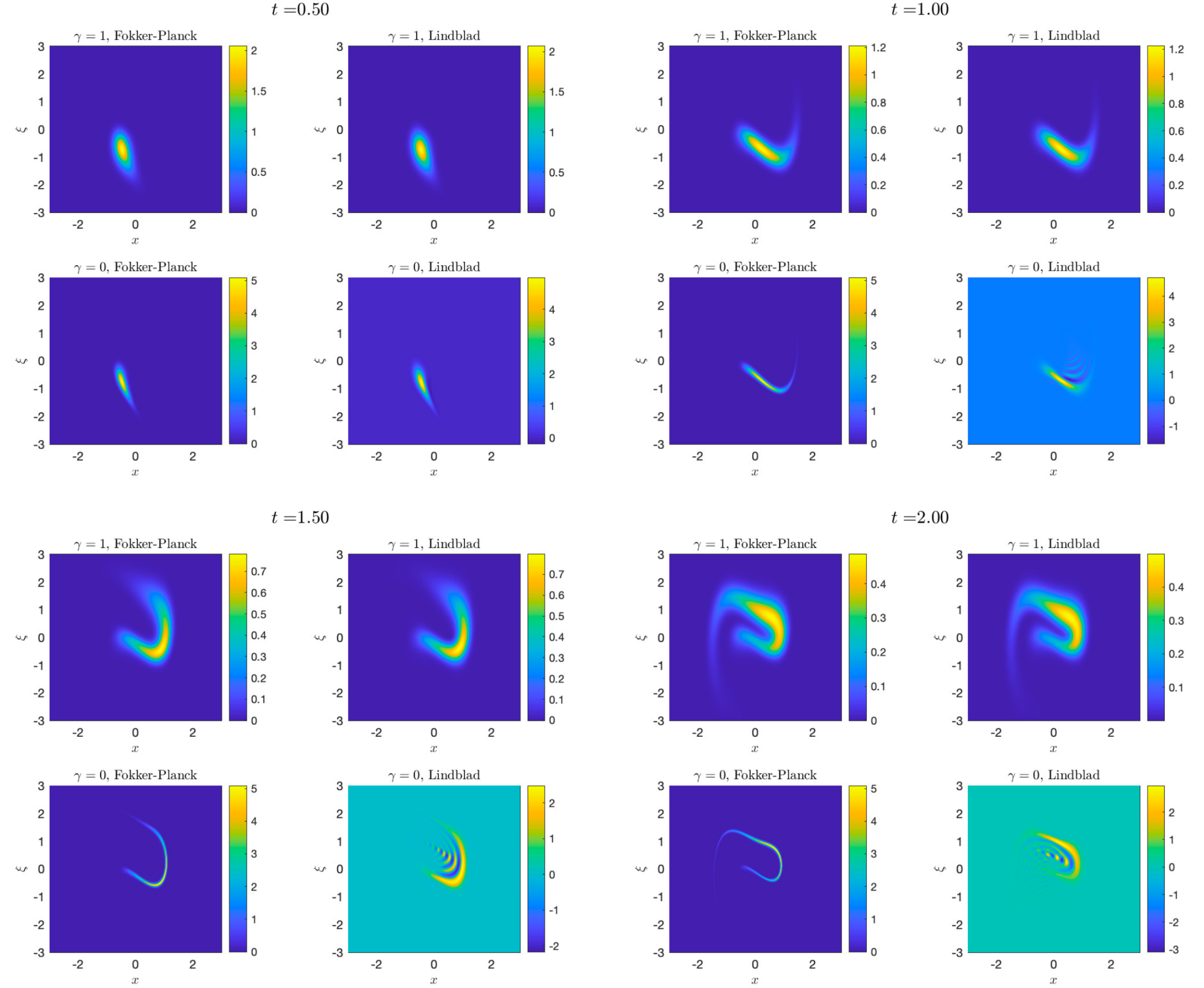}
  \caption{With $P$ as in~\eqref{eq:defV} and $L_j$ as in~\eqref{eq:defL}, we show contour plots of the evolutions of $ a ( t , x, \xi ) $ under the Fokker--Planck equation
  \eqref{eq:FP} and the evolution of $ b ( t, x, \xi ) $, where $ A (t ) = b^{\rm{w}} ( t , x, h D_x ) $ with $A(t)$ evolving under the Lindblad equation~\eqref{e:Lindblad}: i.e. $ b $ is
  the full Weyl symbol of the operator $ A ( t )$.
  The initial data is given by a coherent state \eqref{eq:a0} with $ h_0 = h = 2^{-4} $. 
  The breakdown of the classical/quantum correspondence is clearly visible when $ \gamma = 0$ 
  (Schr\"odinger evolution \eqref{eq:Schr}) while a very good agreement is seen for $ \gamma = 1$. For an animated version, see \url{https://math.berkeley.edu/~zworski/Lin_vs_FP.mp4} and for a description 
  of the numerical schemes, Appendix B.
  \label{f:3}}
\end{figure}

\subsection{Assumptions on $P $ and $ L_j $ and Fokker--Planck evolution} In this paper we will consider \eqref{e:Lindblad} with $P$ and $ L_j$'s given by {\em pseudodifferential 
operators} (see \eqref{eq:quant} for the notation $a^{\rm w}(x,hD)$), that is semiclassical quantizations of classical observables,  satisfying the following assumptions:
\begin{equation}
\label{eq:assa}
\begin{gathered} 
P = p^{\rm{w}} ( x, hD ) ,  \ \ \  |\partial^\alpha p ( x, \xi ) | \leq C_\alpha,  \ \ | \alpha | \geq 2, \ \ \ p = \bar p, \\
L_j = \ell_j^{\rm{w}} ( x, hD ) , \ \ \ |\partial^\alpha \ell_j ( x, \xi ) | \leq C_\alpha,  \ \ | \alpha | \geq 1 , 
 \ \ 1 \leq j \leq J .
\end{gathered}
\end{equation}

If in \eqref{e:Lindblad}, $ A = a^{\rm{w}} ( x, h D) $, then the leading part of the semiclassical 
expansion of $ \mathcal L A $ (see the derivation in \S \ref{s:cla}) is given by 
the action, $ Q a $, of the following Fokker--Planck operator 
\begin{equation}
\label{eq:FP}
Q:=H_p+\frac{\gamma}{2i}\sum_{j=1}^J (2 \{\ell_j,\bar{\ell}_j\} -\ell_j H_{\bar{\ell}_j}+\bar{\ell}_jH_{\ell_j}) +\frac{h\gamma}{4}\sum_{j=1}^J(H_{\bar{\ell}_j}H_{\ell_j}  +H_{\ell_j}H_{\bar{\ell}_j}) .
\end{equation}
Here, $ H_f := \sum_{ j=1}^n \partial_{\xi_j} f \partial_{x_j } - \partial_{x_j} f \partial_{\xi_j } $ is the 
Hamiltonian vector field of $ f = f ( x, \xi ) $, and $ \{ f, g \} := H_f g $ is the Poisson bracket. We note that $ H_p $ is anti-selfadjoint with respect to the standard measure on $ \mathbb R^n \times \mathbb R^n $. 
Since 
\begin{equation}
\label{eq:Bj} \begin{gathered} 
\tfrac 1{ 2i} \left( 2\{\ell_j,\bar{\ell}_j\} -\ell_j H_{\bar{\ell}_j } +\bar{\ell}_jH_{\ell_j} \right)=: \tfrac1i \{ \ell_j, \bar \ell_j \} + B_j , \ \ \ B_j^* = - B_j ,
\end{gathered}\end{equation}
the self-adjoint contribution to the second term is given by  the real valued function
\begin{equation}
\label{eq:fric}
\mu := \frac 1 {2i} \sum_{j=1}^J  \{ \ell_j, \bar \ell_j \} .
\end{equation}
It is interpreted as  friction.
Finally, the last term in \eqref{eq:FP} is self-adjoint and non-negative.  Assumptions \eqref{eq:assa} show that $ 
\mu $ is bounded ($\mu \in S(1) $ in the notation of \S \ref{s:sym}). 

\noindent
{\bf Example.} Suppose $ J = 2n $ and $ \ell_j = x_j $, 
$\ell_{ j + n } = \xi_j $  for $ j \leq n $. Then 
\begin{equation}
\label{eq:exam} Q = H_p + \tfrac12 \gamma h ( \Delta_x + \Delta_\xi ) . 
\end{equation}

When $ \gamma = 0 $  (that is,  when we consider \eqref{eq:Schr}) classical quantum correspondence in the evolution is described using 
{\em Egorov's theorem} -- see \cite[Theorem 11.12, \S 11.5]{z12} and references given there. Here we present it slightly differently, using the Hilbert--Schmidt norm of the operator
-- see Theorem \ref{t:eg} for a general version.
For the evolution \eqref{eq:Schr} with $ A ( 0) =  ( 2 \pi h)^{n/2} a_0^{\rm{w}} ( x , h D)$ where $a_0\in C_c^\infty(\mathbb{R}^{2n})$ is $h$-independent (so that $ \|A (0) \|_{\mathscr L_2 } = \| a_0 \|_{L^2} $) we have
\begin{equation}
\label{eq:eg1} 
 \begin{gathered} 
\|A(t)-\Op((\exp t H_p)^* a_0)\|_{\operatorname{\mathscr{L}_2}}\leq C e^{3 \Gamma t }  h^{2} , 
\end{gathered}
\end{equation}
where $ \| \bullet \|_{\mathscr L_2} $ denotes the Hilbert--Schmidt norm and 
\begin{equation} 
\label{eq:Lyap} \Gamma := \sup_{ |\alpha| = 2 }  \sup_{\mathbb{R}^{2n}}|\partial^\alpha p ( x, \xi ) | .
\end{equation}
 For a more precise 
version of $ \Gamma $, under additional assumptions on $ p $, in terms of Lyapunov exponents of the flow of $H_p $ see 
\cite[Appendix C]{obw} and references given there. For a relation between \eqref{eq:Lyap} and the flow 
see Lemma \ref{l:G2phi}. 

The estimate \eqref{eq:eg1} is not optimal, but as  $\|a^{\rm w}(x,hD)\|_{
\mathscr L_2 } = ( 2 \pi h)^{-n/2} \| a \|_{ L^2 ( \mathbb R^{2n} )} $, 
\eqref{eq:eg1} indicates the basic principle that the agreement with classical evolution breaks down at times proportional to $ \log (1/h) $, the Ehrenfest time.

Motivated by recent papers \cite{hrr, hrrb} by Hern\'andez--Ranard--Riedel we consider the question of an agreement with classical evolution for much longer times: the quantum evolution is given by $ e^{ t \mathcal L} $
where $ \mathcal L $ is the Lindblad operator \eqref{e:Lindblad} and the classical evolution by $ e^ {t Q }$,
where $ Q $ is the Fokker--Planck operator \eqref{eq:FP}. The results are shown in Theorem \ref{t:1} for the special case of $ h$-independent symbols, and in Theorem \ref{t:2} for the more general situation of initial condition in exotic symbol classes.  We show that agreement holds in Hilbert--Schmidt norms. 
The main advantage lies in an easy characterization of Hilbert--Schmidt pseudodifferential operators and in the simplicity of $ L^2$ estimates for the Fokker--Planck evolution defined using \eqref{eq:FP}. 

\noindent
{\bf Remark.} As was shown by Davies \cite{dav}, 
the operator of the form \eqref{e:Lindblad} generates a positivity preserving contraction on the Banach space of self-adjoint trace class operators provided that 
\[ Y := i P - \tfrac12 \sum_{j=1}^J L_j^* L_j \]
is the infinitesimal generator of a strongly continuous one parameter contraction semigroup on 
$ L^2 ( \mathbb R^n ) $. In our case, this follows from the Hille--Yosida theorem and Proposition \ref{p:spec}
(see the proof of Proposition \ref{p:HY} for a similar argument with $ \mathscr L_2 $ playing the role of $ L^2$).

As in \cite{hrr,hrrb} we make a strong non-degeneracy assumption: 
\begin{equation} 
\label{eq:nondeg}
\mathbf H \mathbf H^* \geq c I_{  \mathbb C^{2n} }  , 
 \ \ \ \ 
\mathbf H := [ H_{\ell_1} , \cdots , H_{ \ell_J}, H_{\bar \ell_1} , \cdots , H_{ \bar \ell_J} ] \in M_{ 2n \times 2J } ( \mathbb C ) .
\end{equation}
This cumbersome looking condition corresponds to ellipticity of the second order operator appearing in the classical Fokker--Planck
equation \eqref{eq:FP}  corresponding to \eqref{e:Lindblad} -- see example \eqref{eq:exam} and Remark 5 after Theorem \ref{t:1}.
We also need a more technical condition
\begin{equation}
\label{eq:tech}
|\partial^\alpha \Im \ell_j||\ell_j |+|\Im \ell_j||\partial^\alpha \ell_j|\leq C_\alpha, \ \ |\alpha| \geq 2. 
\end{equation}

\subsection{Lindblad propagation for $h$-independent observables}
With this notation in place we have a special case of Theorem \ref{t:2} in \S \ref{s:aqd}:
\begin{theo}
\label{t:1}
Suppose that $ \mathcal L $ is given by \eqref{e:Lindblad}, assumptions \eqref{eq:assa}, \eqref{eq:nondeg}, and \eqref{eq:tech} hold and $ h^{\frac 13}\leq\gamma \leq h^{-1} $. 
If $a_0\in C_c^\infty(\mathbb{R}^{2n})$ is $h$-independent and $A(t)$ satisfies
\begin{equation*}
\partial_t A ( t) = \mathcal L A ( t ) , \  \  A(0) = ( 2 \pi h )^{n/2} a_0^{\rm{w}} ( x, h D ) , \ \ 
\| A ( 0 ) \|_{ \mathscr L_2 } = \| a_0 \|_{ L^2 ( \mathbb R^{2n}) }, 
\end{equation*}
then for some constant $ C $, 
\begin{equation}
\label{eq:t21noGauss}
\begin{gathered}
\|A(t)- a(t)^{\rm{w}} ( x, h D ) \|_{\operatorname{\mathscr{L}_2}}\leq C e^{(M_0 + C_0 h ) \gamma t} t h^{\frac{1}{2}}\gamma^{-\frac{3}{2}} (1+\gamma)
 (1+ t\gamma^{\frac 32} h^{\frac12}  ),
\end{gathered}
\end{equation}
where 
$$
 (\partial_t-Q)a(t)=0,\ \ \ a(0)= ( 2 \pi h)^{n/2} a_0 ,\qquad M_0:=\sup \mu.
$$
When $ \mu \equiv 0 $ (see \eqref{eq:fric}) then \eqref{eq:t21noGauss} improves to 
\begin{equation}
\label{eq:nofric0}
\begin{gathered}
\|A(t)- a(t)^{\rm{w}} ( x, h D ) \|_{\operatorname{\mathscr{L}_2}}\leq C e^{C h^2 \gamma t }
th^{\frac{1}{2}}\gamma^{-\frac{3}{2}}(1+\gamma).
\end{gathered}
\end{equation}
\end{theo}

\begin{figure}
\centering
\includegraphics[width=14cm]{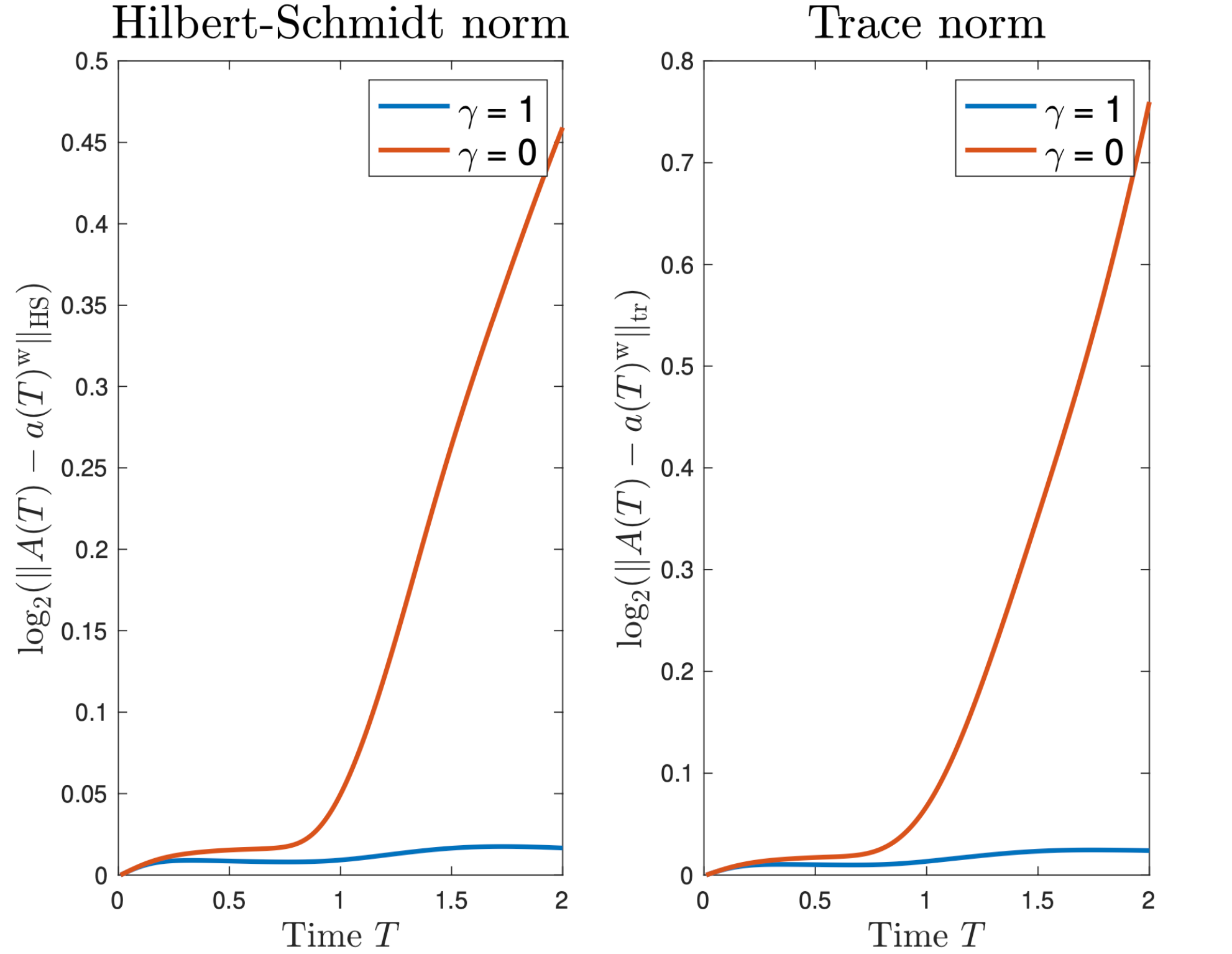}
          \caption{A comparison of classical evalutions with and without damping in the setting
          of Theorem \ref{t:1} with $P$ as in~\eqref{eq:defV} and $L_j$ as in~\eqref{eq:defL}. We choose
       a fixed $ a_0 $ given by \eqref{eq:a0} with           
          $h_0=2^{-3}$ and take $h=2^{-5}$. More quantitative results are presented in 
          Appendix B but the striking difference between the evolutions is clearly visible.
          \label{f:1}}
  \end{figure}

\noindent
{\bf Remarks.} 1. When in \eqref{eq:assa} $ p $ is quadratic and $ \ell_j $'s are linear than the agreement 
of the two evolutions is exact. This corresponds to the same phenomenon in the case of Egorov's theorem -- see
\cite[Theorem 11.9]{z12}. 

\noindent 2. When $ p ( x, \xi ) $ is confining (for instance $ p ( x, \xi ) \geq  |x|^2 + |\xi|^2 $ 
and subharmonic outside of a compact set) then Proposition \ref{p:confined} shows that in Example \eqref{eq:exam} (and most likely in 
greater generality),   $ \| a^{\rm{w}} ( t )  \|_{ \mathscr L_2 }  \geq \| a^{\rm{w}} ( 0) \|_{\mathscr L_2 } /C $ for $ t \leq h^{-\nu } $, $ \nu > 0 $. That means that
the estimates  \eqref{eq:t21noGauss} and \eqref{eq:nofric0} are meaningful for long times.

\noindent
3. To see the reason for the powers of $ h$, $ \gamma $, and $ t $ in \eqref{eq:t21noGauss} consider the simplest case given in 
\eqref{eq:exam}. The classical (Fokker--Planck) evolution is then
\[ ( \partial_t - H_p - \varepsilon^2 \Delta_{x,\xi} ) a ( t ) = 0 , \  \ \ \varepsilon := \sqrt { \gamma h/2 } , 
  \ \ \ a (t )= a (t,x,\xi).  \]
The solutions satisfy the following estimate (immediate if $ H_p = 0 $), see Proposition \ref{p:advect}
(see \eqref{eq:estU0}):
\begin{equation}
\label{eq:est0}   \sum_{ | \alpha | \leq k } \| ( \varepsilon \partial_{x,\xi} )^\alpha a (t) \|_{L^2_{x,\xi} } \leq C \sum_{ | \alpha | \leq k }
\| ( \varepsilon \partial_{x,\xi} )^\alpha a (0) \|_{L^2_{x,\xi} } . \end{equation} 
The key fact is that there is \emph{no} dependence on $ t $ -- that is not the case for the evolution by $ H_p $ alone, see \eqref{eq:rhots}. The composition formula for pseudodifferential operators in Lemma \ref{l:compose}
shows that $ \mathcal L \left[ a(t)^{\rm{w}} (x, h D )\right]  \equiv ( Q a ( t ) )^{\rm{w}}( x, h D )  $ modulo terms quantizing functions bounded by 
the size of $ (1+\gamma)h^2 \partial^3 a (t ) $. These can be estimated using \eqref{eq:est0} where in the case of 
\eqref{eq:cohde} and for $ | \alpha | = 3$, 
\begin{equation}
\label{eq:h2} 
\begin{split} (1 + \gamma ) h^2 \|   \partial_{x,\xi} ^\alpha  a (t) \|_{L^2_{x,\xi} } &\leq 
C (1 + \gamma ) h^2 \varepsilon^{-3} \sum_{ | \alpha | \leq 3 }
\| ( \varepsilon \partial_{x,\xi} )^\beta a (0) \|_{L^2_{x,\xi} } \\
& \leq ( 2 \pi h )^{\frac n2}
C ( 1 + \gamma ){\gamma^{-\frac32}  }
h^{\frac12 }  .
\end{split} \end{equation}
To get \eqref{eq:nofric} we write
\[ A(t) - a(t)^{\rm{w}} (x, h D )  = \int_0^{t} e^{ (t-s) \mathcal L } (  
   \mathcal L  a(t)^{\rm{w}} (x, h D ) - ( Q a ( t ) )^{\rm{w}}( x, h D ) ) ds , \]
which together with \eqref{eq:h2} and the fact that $\|A ( 0 ) \|_{\mathscr{L}_2}=\|a_0\|_{L^2}$, gives \eqref{eq:nofric0}. The extra growth in~\eqref{eq:t21noGauss} results from friction which is absent in this example. We used here the fact that in the example
$ e^{t \mathcal L} $ is a contraction -- in general there could be exponential growth produced by the 
friction term; this is reflected by the exponential prefactor in~\eqref{eq:t21noGauss}. 

\noindent
4. The class of operators $ P$ satisfying \eqref{eq:assa} includes Schr\"odinger operator whose classical dynamics 
exhibits chaotic behaviour. In that case  one expects optimality of $ t \sim \log (1/h) $
limit for classical--quantum correspondence for \eqref{eq:Schr}. For instance we could take
\[ p ( x, \xi ) = \xi_1^2 + \xi_2^2 + x_1^2 + x_2^2  + \lambda ( x) ( x_1^2 x_2 - \tfrac 13 x_2^3) , 
\] 
where $ \lambda \in C_{\rm{c}}^\infty ( \mathbb R^2 ; [ 0, 1]) $ 
and $ \lambda = 1 $ near $0 $. 

\noindent
5. Compared to the models used in the physics literature -- see Unruh--Zurek \cite{unzu} for 
the pioneering discussion of the classical/quantum correspondence for open systems -- 
the ellipticity hypothesis \eqref{eq:nondeg} made in \cite{hrr} and here is too strong. Rather than \eqref{eq:exam}, one
should consider $ \ell_j = x_j $, $ 0 \leq j \leq J = n $ so that the Fokker--Planck operators is given by 
$ Q = H_p + \frac12 \gamma h \Delta_\xi $. This would require more subtle subelliptic estimates (see
Smith \cite{smith} for a recent treatment with an asymptotic parameter) than \eqref{eq:est0}. 
Gong--Brumer \cite{gob} showed numerically that for such operators with chaotic classical dynamics for
$ p$,  
the classical/quantum correspondence persists for long times.

\subsection{Lindblad propagation for mixtures of Gaussian states}

We now state a special case of our theorem where we consider mixtures of Gaussian states in a sense similar to that in~\cite{hrr}. For that we define the standard coherent states:
\[ \psi_{(x_0, \xi_0 )} = ( 2 \pi h )^{-\frac n 4 } e^{ - (x-x_0)^2/2h} e^{ i (x-x_0 ) \xi_0/h }, 
\ \ \ \| \psi_{(x_0, \xi_0 )} \|_{L^2(\mathbb{R}^n)} = 1. \]
The corresponding density operator is 
\begin{equation}
\label{eq:cohde}   
\begin{gathered} A_{(x_0,\xi_0)}  u :=  \psi_{(x_0,\xi_0)} \langle u , \psi_{(x_0,\xi_0)} \rangle, 
\ \ \ \ A_{(x_0,\xi_0)} = a_0^{\rm{w}}( x, h D) , \\
a_{(x_0,\xi_0)} ( x, \xi ) = 2^{ n }\exp\left( - \frac 1 {h}  \left( ( x - x_0)^2 + (\xi -\xi_0)^2 \right) \right). 
\end{gathered}
\end{equation}
We note that in our result the Gaussian $  ( 2 \pi h )^{-\frac n 4 } e^{ - (x-x_0)^2/2h}$ 
could be replaced by $ \alpha h^{-n/4} \psi ( (x - x_0)/ \sqrt h ) $ where $ \psi \in \mathscr S ( \mathbb R^{2n} ) $ and $ \alpha = 1/\| \psi \|_{ L^2} $.

 For a probability measure  $\lambda_h$ on $\mathbb{R}^n\times \mathbb{R}^n$ we define the mixture of Gaussian states:
\begin{equation}
\label{e:mixture}   
\begin{gathered} A_{\lambda_h} u :=  \int \psi_{(x_0,\xi_0)} \langle u , \psi_{(x_0,\xi_0)} \rangle d\lambda_h(x_0,\xi_0), 
\ \ \ \ A_{\lambda_h} = a_{\lambda_h}^{\rm{w}}( x, h D) , \\
a_{\lambda_h} ( x, \xi ) = 2^{ n }\int \exp\left( - \frac 1 {h}  \left( ( x - x_0)^2 + (\xi -\xi_0)^2 \right) \right)d \lambda_h(x_0,\xi_0). 
\end{gathered}
\end{equation}
We note that $ \| A_{\lambda_h } \|_{\mathscr L_1 } = 1 $. {For the Hilbert Schmidt norm we calculate
\[  ( 2 \pi h )^{-n}  2^{2n} \int e^{- \frac 1 {h}  ( ( x - x_0)^2 +  ( x - y_0 )^2 + (\xi -\xi_0)^2 + ( \xi - \eta_0)^2) }
d x d \xi  =
e^{ - \frac 1 h ( ( x_0 - y_0 )^2 + ( \xi_0 - \eta_0 )^2) }  , 
 \]
 so that 
  \[ \| A_{\lambda_h } \|_{ \mathscr L_2 }^2 = \int\!\!\!\int \exp \left( - \frac 1 h \left( ( x_0 - y_0 )^2 + ( \xi_0 - \eta_0 \right )^2 \right) d \lambda_h ( x_0, \xi_0 ) d\lambda_h ( y_0, \eta_0 ) . \]
If $ \lambda_h = \mu (x, \xi ) dx d\xi $, where $ \mu $ is smooth and $ h$ independent we are close to 
the case considered in Theorem \ref{t:1} and $\| A_{\lambda_h } \|_{\mathscr L_2 } \sim h^{n/2} $.}

As in \eqref{eq:eg1} when $ \gamma = 0 $, we obtain a version of Egorov's Theorem: for the solution of \eqref{eq:Schr} with $ A ( 0) =  a_{\lambda_h}^{\rm{w}} ( x , h D)$, 
\begin{equation}
\label{eq:eg} 
 \begin{gathered} 
\|A(t)-((\exp t H_p)^* a_{\lambda_h})^{\rm{w}} ( x, hD ) \|_{\operatorname{\mathscr{L}_2}}\leq C e^{3 \Gamma t }  h^{\frac{1}{2}} {\|A(0)\|_{\mathscr{L}_2}}.
\end{gathered}
\end{equation}

On the other hand, for the Lindblad evolution, the quantum-- classical agreement lasts substantially longer as can be seen from the following special case of Theorem~\ref{t:2}
\begin{theo}
\label{t:1Gauss}
Suppose that $ \mathcal L $ is given by \eqref{e:Lindblad}, assumptions \eqref{eq:assa}, \eqref{eq:nondeg}, and \eqref{eq:tech} hold and $ \gamma < h^{-1} $. 
If, in the notation of \eqref{e:mixture}, $A(t)$ satisfies
\begin{equation}
\label{eq:evolGauss} 
\partial_t A ( t) = \mathcal L A ( t ) , \  \  A(0) = A_{\lambda_h} , 
\end{equation}
then for some constant $ C $, 
\begin{equation}
\label{eq:t21}
\begin{gathered}
\|A(t)- a(t)^{\rm{w}} ( x, h D ) \|_{\operatorname{\mathscr{L}_2}}\leq Ce^{(M_0 + C h )\gamma t} t
( \gamma + \gamma^{-\frac{3}{2}} ) 
 h^{\frac12} (1+ t\gamma^{\frac32} h^{\frac12}   ){\|A_{\lambda_h}\|_{\mathscr{L}_2}} , 
\end{gathered}
\end{equation}
where 
$$
 (\partial_t-Q)a(t)=0,\ \ \ a(0)= a_{\lambda_h} ,\qquad M_0:=\sup \mu.
$$
When $ \mu \equiv 0 $ (see \eqref{eq:fric}) then \eqref{eq:t21} improves to 
\begin{equation}
\label{eq:nofric}
\begin{gathered}
\|A(t)- a(t)^{\rm{w}} ( x, h D ) \|_{\operatorname{\mathscr{L}_2}}\leq C e^{C h^2 \gamma t }  t
( \gamma + \gamma^{-\frac{3}{2}} ) 
 h^{\frac12}{\|A_{\lambda_h}\|_{\mathscr{L}_2}}, 
\end{gathered}
\end{equation}
\end{theo}

\begin{figure}
      \centering
                \includegraphics[width=14cm]{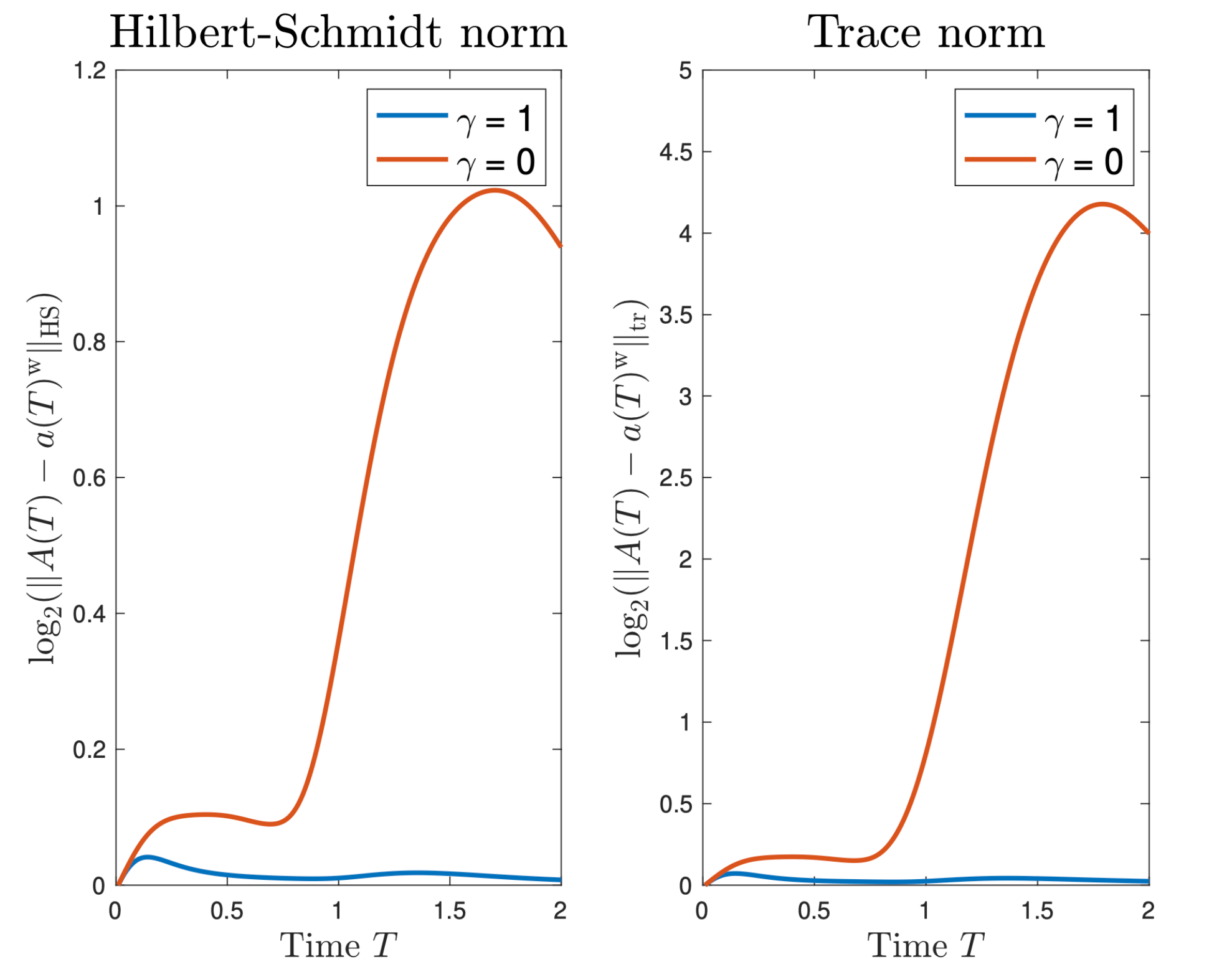}
          \caption{The analogue of Figure \ref{f:1} with the initial condition given 
         by a  coherent state \eqref{eq:a0} with $ h_0 = h = 2^{-5}$. Again the improvement
         in the agreement of quantum and classical evolutions is striking. \label{f:2}}
              \end{figure}

\noindent
{\bf Remark.} The time scales appearing in Theorems \ref{t:1}, \ref{t:1Gauss} and \ref{t:2}
agree with the time scales 
 in \cite{hrr}, as long as $ \gamma \leq 1 $:  Theorem 3.1 there gives the bound $ C \max( 1, \gamma^{-\frac 32})  h^{\frac12} t $  for a two tier comparison of evolution of specially constructed Gaussian states.
Under the assumptions in Theorem \ref{t:1} it reads as
 \begin{equation}
 \label{eq:HRR}   \begin{split}  \| A ( t ) - \tilde a (t)^w ( x, h D ) \|_{ \mathscr L_1} &  \leq C \max( 1, \gamma^{-\frac32})  h^{\frac12} t , \\
h^{-n}   \| \tilde a ( t ) - a ( t) \|_{ L^1 ( \mathbb R^{2n}) } &  \leq  C \max( 1, \gamma^{-\frac32})  h^{\frac12} t , 
 \end{split} \end{equation}
where $ \| \bullet \|_{ \mathscr L_1 } $ is the trace class norm. Remarkably, since the semigroup $ e^{ t 
\mathcal L} $  is contracting  on trace class operators,  there is no exponential growth even when friction is positive. The estimate does {\em not} provide a bound on  $ A ( t ) - a (t)^w ( x, h D )$ in any norm, 
but has the following natural consequence \cite[(1.7)]{hrr}:
\[  \begin{split}  &  \tr \left( A ( t ) b^{\rm{w}} ( x, h D )\right) -  ( 2 \pi h )^{-n} \int_{ \mathbb R^{2n} } a (t, x, \xi )  b ( x, \xi) dx d\xi \\
& \ \ \ \ \ \ \ =
\mathcal O\big( t \max( 1, \gamma^{-\frac32})  h^{\frac12}  \big) ( \| b^{\rm{w}} ( x, h D ) \|_{L^2 \to L^2} + \| b \|_{ L^\infty } ) . \end{split} \]
This is (typically) stronger than the corresponding consequence of \eqref{eq:t21}:
\[  \begin{split} &  \tr \left( A ( t ) b^{\rm{w}} ( x, h D )\right)  -  ( 2 \pi h )^{-n} \int_{ \mathbb R^{2n} }  a (t, x, \xi )  b ( x, \xi)  dx d\xi \\
& \ \ \ \ \ \ \ =
\mathcal O  \big(e^{(M_0 + C h )\gamma t} t
( \gamma + \gamma^{-\frac32} ) 
h^{\frac12} (1+ t\gamma^{\frac32} h^{\frac12} )  \big) ( 2 \pi h )^{-n/2}   \| b \|_{L^2 } . \end{split} \]
We stress, however, that Theorem 
 \ref{t:2} below applies to very general initial states $ A(0)$ of which Gaussian states or their {mixtures} are an example. In addition,  at the cost of further terms in the expansion, it gives approximation of the Lindblad evolution modulo $ \mathcal O(e^{(M_0+Ch)\gamma t}(th^{1/2}\gamma^{-\frac32})^N(1+\gamma)^{N}
 (1+ t\gamma^{\frac32} h^{\frac12 } )^{\frac{N+2}2 }  $ for any $N$. 
 For instance, when 
 $$\gamma = h^{\delta},  \ \ \delta <  \tfrac13, \ \ \ 
t\leq h^{-\nu}, \ \  \nu < \tfrac12 \min (2 \delta, 1- 3 \delta ) , $$ this gives an expansion modulo $\mathcal O(h^\infty)$.

This paper is self-contained except for some basic facts about semiclassical quantization from 
\cite[Chapter 4]{z12}. It is organized as follows. In \S \ref{s:sym} we review the definition of 
pseudodifferential operators and symbol classes. We introduce a new $L^2$-based symbol class
which is natural for the study of Hilbert--Schmidt operators, and show the properties of the
corresponding pseudodifferential calculus. In \S \ref{s:ego} we present a variant of Egorov's theorem 
with Hilbert--Schmidt norm and in \S \ref{s:semiL} we prove mapping properties of $ e^{ t \mathcal L}$.
\S \ref{s:cla} is then devoted to estimates on the Fokker--Planck evolution. A general 
result about agreement of classical and quantum dynamics in Hilbert--Schmidt norm is proved in 
\S \ref{s:aqd}. In \S \ref{s:lindbladBound}, we consider situations where we can effectively control the Hilbert--Schmidt norm of the Lindblad evolution from below. Finally, in the Appendix~\ref{a:app}, we review some properties of pseudodifferential operators with quadratic symbol growth and in Appendix~\ref{a:num} we describe some numerical experiments on the Linblad evolution. 

\medskip

\noindent
{\bf Acknowledgements.}  We would like to thank Simon Becker for pointing out reference \cite{dav}
which clarified the trace class properties of $ e^{ t \mathcal L } $ and Zhenhao Li for helpful comments on early versions of this paper. We are also grateful to Felipe Hern\'andez, Daniel Ranard and  Jess Riedel for an illuminating discussion of the results and methods of \cite{hrr}.

JG acknowledges support
from EPSRC grants EP/V001760/1 and EP/V051636/1 and MZ from the NSF grant
DMS-1901462 and the Simons Foundation under a “Moir\'e Materials Magic” grant.
MZ would like to thank Paul Brumer for introducing him, many years ago in the context of discussing 
\cite{gob}, to the problem of classical/quantum correspondence for Lindbladians. MZ is grateful to 
University College London for the generous hospitality during the writing of this paper.

\section{Symbol spaces and quantization}
\label{s:sym}

The operators introduced in \S \ref{s:intr} are defined using {\em pseudodifferential operators}
which are obtained by a Weyl quantization process: at first for $ a \in \mathscr S ( \mathbb R^{n}_x
\times \mathbb R^n_\xi ) $ (here $ \mathscr S $ denotes Schwartz functions, that is functions $ u$
for which $ x^\alpha \partial^\beta u $ are bounded for all multiindices $ \alpha $ and $ \beta$;
 $ \mathscr S' $ denotes its dual, the space of tempered distributions -- see \cite[Chapter 3]{z12}) we define
\begin{equation}
\label{eq:quant}
 \Op( a) u =  a^{\rm{w}} ( x, h D , h ) u := 
\frac{1}{ (2 \pi h)^n} \int a \left( \frac{x+y} 2 , \xi \right) e^{ \frac i h \langle x- y , \xi \rangle } u ( y ) d y d \xi.
\end{equation}
The Hilbert--Schmidt norm has a clean expression in terms of the symbol $ a $ (PDE parlance for classical observables):
\[ \| \Op ( a ) \|_{\mathscr L_2 }^2 = \tr \Op ( a ) \Op ( a)^* = 
\frac{1}{(2 \pi h)^n } \int_{\mathbb R^{2n} } | a( x, \xi ) |^2 d x d\xi .  \]
This is in contrast with the trace class norm which does not have an easy characterization in terms of $ a $ and its estimates require $L^1 $ norm of derivatives of $ a $ -- 
see \cite[Chapter 9]{DiSj}. 

In this paper we consider different classes of symbols for which \eqref{eq:quant} remains valid 
and has interesting composition properties (as an operator $ \Op ( a ) : \mathscr S \to \mathscr S' $
the operator \eqref{eq:quant} is well defined for $ a \in \mathscr S' ( \mathbb R^{2n} ) $ -- see 
\cite[Theorem 4.2]{z12}). We first recall the standard symbol class: for 
$ m : \mathbb R^{2n} \to [0, \infty ) $ satisfying $ m(z)/m(w) \leq C ( 1  + |z-w|)^N $, 
\begin{equation}
\label{eq:Sdm} a \in S_{\delta} ( m ) \ \Longrightarrow \ | \partial_z^\alpha a ( z, h ) | \leq C_\alpha h^{-\delta|\alpha|} m ( z ) , \ \ \ 
z = ( x, \xi ) \in \mathbb R^{2n} . \end{equation}
When $ \delta = 0 $ we write $ S (m ) $ and when $ m = 1$, $ S_\delta $. 

The next class corresponds to the conditions in \eqref{eq:assa}: for smooth function on $ \mathbb R^{2n} $, 
\begin{equation}
\label{eq:defSk} u (z, h ) \in S_{(k)} \ \Longleftrightarrow \  | \partial^\alpha_z  u ( z , h ) | \leq C_\alpha,  \ \ |\alpha| \geq k , \end{equation}
with constants $ C_\alpha $ independent of $ h $.
The seminorms are given by the best constants $ C_\alpha $. 

In dealings with Hilbert--Schmidt operators it is natural to consider symbols whose bounds are defined
using $ L^2 $ norms. For smooth functions on $ C^\infty ( \mathbb R^{2n} ) $ depending
on parameters $ h $ we define, for $ 0 \leq \rho < 1 $, 
\begin{equation}
\label{eq:SL2} a  \in S^{L^2}_{\rho} \ \Longleftrightarrow \ 
h^{-\frac n2} \| \partial_z^\alpha  a \|_{L^2 ( \mathbb R^{2n} ) }  \leq C_\alpha h^{-\rho |\alpha | } . \end{equation}
with the obvious seminorms.  We note that the Sobolev embedding theorem and an interpolation 
argument show that
$ | \partial_z^\alpha a | \leq C_\alpha'  h^{-\rho (|\alpha | + n + \delta ) + \frac n2} $, for any 
$ \delta > 0 $.  Hence for $ \rho = 0 $
the $ L^2 $ based spaces are contained in $ h^{\frac n2} S ( 1) $ defined above, and in general 
\begin{equation}
\label{eq:L221} S^{L^2}_{ \rho } \subset h^{ - \rho ( n + ) + \frac n2 } S_{
\rho} ( 1 ) . 
\end{equation}

It will also be useful to consider mixed spaces obtained by taking tensor products:
\begin{equation} 
\label{eq:defT} c ( z, w ) \in S \otimes S^{L^2}_{\rho } \ \Longleftrightarrow \ 
h^{-\frac n2}  \| \sup_{ z } \partial^\alpha_z \partial^\beta_w c ( z, \bullet ) \|_{L^2 } \leq C_{\alpha \beta } h^{-\rho|\beta|} , \ \ z , w \in \mathbb R^{2n} . 
\end{equation}
We stress that we always demand that $ 0 \leq \rho < 1 $. 

\noindent
{\bf Remark} Another choice of the norm could be given by $ \sup_z \|  \partial^\alpha_z \partial^\beta_w c ( z, \bullet ) \|_{L^2 } $ and both agree on
products. The choice in definition \eqref{eq:defT} is motivated by the fact that $\| f ( w, w ) \|_{L^2_w} \leq \| \sup_z| f ( z, w )| \|_{L^2_w} $
{which does not work for the other choice}.

For the properties of operators which are quantizations of $ a \in S_\delta ( m ) $ see \cite[Chapter 4]{z12}.
The same methods apply to operators obtained from $ a \in S_{(k)}$ and are reviewed in 
the appendix. In particular we obtain 
spectral properties
of operators quantizing $ S_{(2)} $. Since the properties of $ S_{\rho}^{L^2} $ and
$ 0 \leq \rho < 1 $ are more unusual we present them in this section. We start with 
\begin{lemm}
\label{l:Q2Tens}
Suppose that $ Q : \mathbb R^{2n} \times \mathbb R^{2n } \to \mathbb R $ is a non-degenerate bilinear quadratic form. Then, using definition \eqref{eq:defT}, 
\begin{equation}
\label{eq:Q}
 e^{ i h Q ( D_z, D_w )} :  S \otimes S^{L^2}_{\rho}  \to S  \otimes S^{L^2}_{\rho} ,  
 \end{equation}
is continuous and for every $ N $
\begin{equation}
\label{e:asymptoticQ}
e^{ihQ(D_z,D_w)}a-e^{\frac{i\pi}{4}\sgn Q}\sum_{k=0}^{N-1} \Big(\frac{h}{i}\Big)^k\frac{1}{k!}Q(D_{z},D_{w})^ka(z,w)\in h^{N(1-\rho)}S  \otimes S^{L^2}_{\rho},
\end{equation}
where $ \sgn Q $ is the signature of $ Q $ considered as a quadratic form on $ \mathbb R^{2n} 
\times \mathbb R^{2n} $. 
\end{lemm}
\begin{proof}
We denote by 
$ B $ the symmetric matrix corresponding to our quadratic form: 
$Q(\zeta,\omega)=\frac{1}{2}\langle B(\zeta,\omega),(\zeta,\omega)\rangle.$ For $a\in S (1)
\otimes S^{L^2}_{ \rho} \subset h^{ - \rho ( n + ) + \frac n2 } S( 1 ) \otimes S_{\rho} (1) $, hence the expression 
$$
c(z,w):=  e^{ i h Q ( D_z, D_w )} a(z,w)
$$
makes sense as an element in $\mathscr{S}'$ (to see this, we apply e.g.~\cite[Theorem 4.17]{z12} with for each fixed value of $h$)
and 
 by~\cite[Theorem 4.8]{z12}, for $a\in \mathscr{S}$, 
$$
c ( z, w ) =\frac{|\det B |^{-\frac{1}{2}}}{(2\pi h)^{2n}}\int_{\mathbb{R}^{2n}}\int_{\mathbb{R}^{2n}}e^{\frac{i}{h}\varphi(z_1,z_2)}a(z+z_1,w+w_1)dz_1dw_1,
$$
 where
 $$
 \varphi(z_1,w_1)=-\tfrac{1}{2}\langle B^{-1}(z_1,w_1),(z_1,w_1)\rangle.
 $$
 Since $a\in h^{ - \rho ( n + ) + \frac n2 } S( 1 ) \otimes S_{\rho} (1) $, this integral  can be understood in the sense of oscillatory integrals and defines an element of $\mathscr{S}'$ -- see \cite[\S 3.6]{z12}. Recall also that oscillatory integrals allow for integrations by parts.
 
 Set $v_1=h^{-\rho}w_1$, and $\chi \in C_c^\infty(\mathbb{R}^{2n}\times \mathbb{R}^{2n})$ with $\chi \equiv 1$ near $0$ and $\supp \chi \subset B(0,1)$. Then using the fact that $ w_1 
\mapsto \varphi( z_1, w_1 ) $ is linear, we obtain
 $$
\begin{aligned}
c(z,w)&=\frac{|\det B |^{-\frac{1}{2}}}{(2\pi h^{1-\rho})^{2n}}\int_{\mathbb{R}^{2n}}\int_{\mathbb{R}^{2n}}e^{\frac{i}{h^{1-\rho}}\varphi(z_1,v_1)}a(z+z_1,w+h^{\rho}v_1)dz_1dv_1,\\
&=\frac{|\det B |^{-\frac{1}{2}}}{(2\pi h^{1-\rho})^{2n}}\int_{\mathbb{R}^{2n}}\int_{\mathbb{R}^{2n}}e^{\frac{i}{h^{1-\rho}}\varphi(z_1,v_1)}\chi(z_1,v_1)a(z+z_1,w+h^{\rho}v_1)dz_1dv_1\\
& \qquad  + \frac{|\det B |^{-\frac{1}{2}}}{(2\pi h^{1-\rho})^{2n}}\int_{\mathbb{R}^{2n}}\int_{\mathbb{R}^{2n}}e^{\frac{i}{h^{1-\rho}}\varphi(z_1,v_1)}(1-\chi(z_1,v_1))a(z+z_1,w+h^{\rho}v_1)dz_1dv_1\\
&=:c_1(z,w)+ c_2(z,w)
\end{aligned}
$$
We start by considering $c_1$. In this case, the integrand is compactly supported and we may apply the method of stationary phase~\cite[Theorem 3.16 and Theorem 3.17]{z12}. That gives
\begin{align*}
&\Big| \partial_{z}^{\alpha_1}\partial_w^{\alpha_2}\Big(c_1(z,w)-e^{\frac{i\pi}{4}\sgn B}\sum_{k=0}^{N-1} \Big(\frac{h^{1-\rho}}{i}\Big)^k\frac{1}{k!}\Big(Q(D_{z_1},D_{v_1})^ka(z+z_1,w+h^{\rho}v_1)|_{z_1=v_1=0}\Big)\Big)\Big|\\
&\leq C_Nh^{(1-\rho)N}\sum_{|\beta_1|+|\beta_2|\leq 2N+4n+1}h^{-\rho|\alpha_2|}\sup_{|(z_1,v_1)|< 1}|\partial_{z_1}^{\beta_1+\alpha_1}\partial_{v_1}^{\beta_2+\alpha_2}a(z+z_1,w+h^{\rho}v_1)|\\
&=:C_Nh^{(1-\rho)N}\sum_{|\beta_1|+|\beta_2|\leq 2N+4n+1}R_{\alpha\beta}(z,w), 
\end{align*}
with the estimates on the remainder provided by Sobolev's embedding:
\begin{align*}
|R_{\alpha\beta}(z,w)|^2&\leq h^{-2\rho|\alpha_2|}\sum_{\gamma\leq 2n+1}\|\partial_{z_1}^{\beta_1+\alpha_1+\gamma_1}\partial_{v_1}^{\beta_2+\alpha_2+\gamma_2}a(z+\cdot,w+h^{\rho}\cdot)\|_{L^2(B_{ \mathbb R^{2n}}(0,1))}^2.
\end{align*}
Hence, with $ B := B_{ \mathbb R^{4n} } (0,1)$,
\begin{align*}
\int_{\mathbb{R}^{2n} }\sup_z |R_{\alpha\beta}(z,w)|^2 dw
&\leq h^{-2\rho|\alpha_2|}\sum_{\gamma\leq 2n+1}\int_{\mathbb{R}^{2n}}\int_{B}\sup_z|\partial_{(z_1,v_1)}^{\alpha+\beta+\gamma}a(z+z_1,w+h^{\rho}v_1)|^2dz_1dv_1dw\\
 &\leq   \int_{B}h^{2\rho(|\gamma_2|+|\beta_2|)}\|\sup_z|\partial_{(z,w)}^{\alpha+\beta+\gamma}a(z,\cdot)|\|_{L^2}^2dz_1dv_1\\
& \leq  C  \sum_{\gamma\leq 2n+1}h^{2\rho(|\gamma_2|+|\beta_2|)}\|\sup_z|\partial_{(z,w)}^{\alpha+\beta+\gamma}a(z,\cdot)|\|_{L^2}^2.
\end{align*}
In particular, this implies that 
$$
c_1(z,w)-e^{\frac{i\pi}{4}\sgn B}\sum_{k=0}^{N-1} \Big(\frac{h^{1-\rho}}{i}\Big)^k\frac{1}{k!}\Big(Q(D_{z_1},D_{v_1})^ka(z+z_1,w+h^{\rho}v_1)\Big)|_{v_1=z_1=0}$$
is in $ h^{(1-\rho)N}S \otimes S^{L^2}_{\rho} $.

We now consider the remaining term in $ c$, $ c_2 $, and note that on $\supp (1-\chi)$, $|\partial_{(z_1,v_1)}\varphi(z_1,v_1) |\geq c\langle (z_1,v_1)\rangle $. Hence, integration by parts 
(justified by the definition of the oscillatory integral) yields, for $N>2n+1$,
\begin{align*}
&\ \ \ \ \ \ \ \ h^{(\rho-1)(2N-4n)}\|\sup_z\partial_{(z,w)}^\alpha c_2(z,\cdot)\|^2_{L^2}\\
&\leq C_N \int \sup_z\Big(\iint\sum_{|\beta_1|+|\beta_2|\leq N}\langle (z_1,v_1\rangle^{-N}\Big|\partial_{(z,w)}^\alpha \Big(\partial_{z}^{\beta_1}(h^{\rho}\partial_{w})^{\beta_2}a(z+z_1,w+h^{\rho}v_1)\Big)\Big|dz_1dv_1\Big)^2dw\\
&\leq C_N \iiint\sum_{|\beta_1|+|\beta_2|\leq N}\langle (z_1,v_1)\rangle^{-2N+2n+1}\sup_z\Big| \partial_{(z,w)}^\alpha \partial_{z}^{\beta_1}(h^{\rho}\partial_{w})^{\beta_2}a(z+z_1,w+h^{\rho}v_1)\Big|^2dz_1dv_1dw\\
&\leq  C_N \iint\sum_{|\beta_1|+|\beta_2|\leq N}\langle (z_1,v_1)\rangle^{-2N+2n+1}\big\|\sup_z \partial_{(z,w)}^\alpha \partial_{z}^{\beta_1}(h^{\rho}\partial_{w})^{\beta_2}a(z,\cdot)\big\|_{L^2}^2dz_1dv_1\\
&\leq C_N \sum_{|\beta_1|+|\beta_2|\leq N}\big\|\sup_z \partial_{(z,w)}^\alpha \partial_{z}^{\beta_1}(h^{\rho}\partial_{w})^{\beta_2}a(z,\cdot)\big\|_{L^2}^2 .
\end{align*}
Hence, we have  $c_2\in h^{(N-2n)(1-\rho)}S \otimes S^{L^2}_{\rho} $ for arbitrary $N$ and 
$ c\in  S  \otimes S^{L^2}_{\rho}$.  The argument also shows that that the map from $a$ to $c$ is continuous, and~\eqref{e:asymptoticQ} holds.
\end{proof}

We can write the composition law for operators in $S^{L^2}_{\rho}$ with $S_{(k)}.$
\begin{lemm}
\label{l:compose}
Let $0\leq \rho <1$, $k\geq 0$, $a\in S_{(k)},b\in S^{L^2}_{\rho}$. Then,  
$$
\Op(a)\Op(b)=\Op(c),
$$
where $ c $ has the following expansion: for $N\geq k$,
\begin{equation}
\label{e:L2Comp}
c(x,\xi)- \sum_{j=0}^{N-1} \frac{1}{j!} \left(\frac{  h} {2i }\sigma(D_x,D_\xi,D_y,D_\eta)^ja(x,\xi)b(y,\eta)\right)|_{\substack{y=x\\\eta=\xi}}\in h^{N(1-\rho)}S^{L^2}_{\rho}.
\end{equation}
\end{lemm}
\begin{proof}
Writing $ z = ( x, \xi ) $, $ w = ( y , \eta ) $, we have
$$
\Op(a)\Op(b)= \Op(c),\qquad c(z):= e^{ihA(D_{z,w})}a(z)b(w)|_{z=w},
$$
where 
$
A ( D_{z,w} ):= - \tfrac12 \sigma ( D_x, D_\xi , D_y, D_\eta ) .
$
By Taylor's formula
\[
 c  ( z , h ) = \sum_{ \ell=0}^{N-1} \frac{1}{\ell!} ( i h A ( D ) )^\ell ( a ( z ) b ( w) )|_{
z = w} + R_N ( z , h ) \]
where 
\[ R_N ( z, h ) : \frac{1}{(N-1)!} \int_0^1 ( 1 - t)^{N-1} e^{ i t h A ( D ) } ( i h A ( D ) )^{N } ( 
( a ( z ) b ( w) )|_{z, w} dt . \]
For  $N\geq k$,
$$
A(D_{z,w})^N a(z)b(w)\in h^{-N\rho}S  \otimes S^{L^2}_{\rho}.
$$
Hence, Lemma~\ref{l:Q2Tens} applies and $ e^{ i h t A ( D ) } :S  \otimes S^{L^2}_{\rho}  \to S  \otimes S^{L^2}_{\rho}$ has uniform bounds in $t\in[0,1]$. 
Now, for $e\in S  \otimes S^{L^2}_{\rho}$, we have
$$
\|\partial_{w}^\alpha e(w,w)\|_{L^2}\leq C \sum_{|\beta|\leq |\alpha|}\|\partial_{(z,w)}^{\beta}e(z,w)|_{w=z}\|_{L^2_w}\leq C \sum_{|\beta|\leq |\alpha|}\|\sup_{z}|\partial_{(z,w)}^{\beta}e(z,\cdot)|\|_{L^2}.
$$
We conclude that 
$ R_N \in h^{(1-\rho)N} S^{L^2}_{\rho}  $ which is \eqref{e:L2Comp}.
\end{proof}

\section{Egorov's theorem revisited}
\label{s:ego} 

We give a variant of Egorov's theorem which is analogous to Theorems \ref{t:1} and \ref{t:2} and 
uses propagation of quantum observables in symbol classes $ S_\rho^{L^2} $ introduced in \S \ref{s:sym}. In fact, the proof
of Theorem \ref{t:2} follows the same strategy with improved estimates coming from diffusion estimates: Lemma \ref{l:G2phi} below (see also \eqref{eq:rhots}) is replaced by Proposition \ref{p:advect}.

We start with a lemma relating the constant $ \Gamma $ in \eqref{eq:Lyap} to the properties of the flow (see \cite[Lemma 11.11]{z12} for a slightly different version)
\begin{lemm}
\label{l:G2phi}  
Let $ \varphi_t := \exp t H_p $ where $ p $ satisfies \eqref{eq:assa}. Then 
\begin{equation} 
\label{eq:estph}
| \partial^\alpha \varphi_t ( x, \xi) |_{\ell^\infty ( \mathbb R^{2n} ) } \leq  C_{\alpha } 
e^{ \Gamma |\alpha| t } ,  \ \ \alpha \in \NN^{2n} , \ \ | \alpha |
> 0 . 
\end{equation}
\end{lemm}

In the proof of Lemma \ref{l:G2phi} we use the following version of Gr\"onwall's inequality:
\begin{lemm}
\label{l:gron}
Let $\Gamma \in \mathbb{R}$ and suppose that $u:\mathbb{R}\to \mathbb{R}$ is continuous and satisfies
\begin{equation}
\label{e:u}
u(t)\leq v(t)+\Gamma\int_0^tu(s)ds.
\end{equation}
Then,
$$
u(t)\leq v(t)+\Gamma \int_0^te^{\Gamma (t-s)}v(s)ds,  \ \ \ t\geq 0.
$$
\end{lemm}
\begin{proof}
Define
$
w(t):=\int_0^t u(s)ds.
$
Then, $w$ is continuously differentiable and satisfies
$$
w'(t)\leq v(t)+\Gamma w(t), \ \ \ w(0) = 0 .
$$
Hence, conjugating by $e^{-\Gamma t}$ and integrating gives 
$$
w(t)\leq \int_0^te^{\Gamma (t-s)}v(s)ds,
$$
which, after substitution in~\eqref{e:u}, finishes the proof.
\end{proof}

\begin{proof}[Proof of Lemma \ref{l:G2phi}]
The proof of \eqref{eq:estph} is an  induction on  $ | \alpha | $. The first step is the case of $ |\alpha | =1$. Since 
$ (d /{ dt } ) \varphi_t  = H_{p} ( \varphi_t  )$, 
\begin{equation}
\label{eq:first} \frac d { dt }  \left( \partial^\alpha \varphi_t \right)  
=  \partial H_{p} ( \varphi_t  )  \partial^\alpha \varphi_t  ,  \ \ \  \partial^\alpha \varphi ( 0 ) =  \alpha . \end{equation}
Since the entries of the matrix  $ \partial H_{p} $ are bounded by $ \Gamma $, integration gives
\[ \sup_{ \mathbb R^{2n} } |\partial^\alpha \varphi_t |_{\ell^\infty} \leq 1 + \Gamma \int_0^t \sup_{ \mathbb R^{2n} } |\partial^\alpha \varphi_s |_{\ell^\infty} ds . \]
Lemma~\ref{l:gron} then gives \eqref{eq:estph} for $ |\alpha | = 1 $.

Now assume $|\alpha| = \ell $ and suppose the estimate \eqref{eq:estph} is
valid for all multiindices $ \beta $ with $ 1 \leq |\beta| < \ell$.  We differentiate \eqref{eq:first}, to find
\begin{equation}
\label{eq:faada}
\frac d { dt }  \left( \partial^\alpha \varphi_t \right)  
=  \partial H_{p} ( \varphi_t  )  \partial^\alpha \varphi_t  
+ g(t) , 
\end{equation}
where $ g(t)  $ is a sum of terms having the form 
\[ 
g_{ \alpha \beta } \circ \varphi_t  \,  \partial^{\beta_1} \varphi_t \cdots  \partial^{\beta_k} \varphi_t, \ \ \  g_{\alpha \beta } \in S ( 1 ) , 
  \]   
for $\beta_1 + \cdots \beta_k = \alpha$ and $0 < | \beta_j | < |\alpha| = \ell$
 $(j = 1, \dots, k)$.
The induction hypothesis implies
$  \sup_{ \mathbb R^{2n} }| g ( t )   |_{\ell^\infty}  \leq C e^{ \Gamma | \alpha | | t | } $. Integrating as above, we obtain
\[ \sup_{ \mathbb R^{2n} } |\partial^\alpha \varphi_t |_{\ell^\infty} \leq Ce^{\Gamma |\alpha|t} + \Gamma \int_0^t \sup_{ \mathbb R^{2n} } |\partial^\alpha \varphi_s |_{\ell^\infty} ds . \]
 and we can use Lemma~\ref{l:gron} 
to obtain \eqref{eq:estph}. 
\end{proof}

\begin{theo}
\label{t:eg}
Suppose that $ \mathcal L_0 $ is given by \eqref{eq:Schr} with $ P $ satisfying \eqref{eq:assa} and $0\leq \rho<\frac{2}{3}$. 
If   $A(t)$ satisfies (in the notation of \S \ref{s:sym})
\begin{equation*}
\partial_t A ( t) = \mathcal L_0 A ( t ) , \  \  \ A(0) = \Op (a_0 ) , \ \  a_0 \in S^{L^2}_{\rho} ,  
\end{equation*}
Then, for every $ N $ there exist $C_N>0$ and $ a ( t ) \in S_{\rho(t)}^{L^2} $ such that 
for $ \Gamma $ given by \eqref{eq:Lyap} and 
\begin{equation}
\label{eq:rhot}   \rho(t) := \rho +  \frac{ \Gamma t}{ | \log h | } \leq \frac{2}{3}, \end{equation}
$  a(t)- (\exp t H_p)^* a_0  \in h^{2-3 \rho}  e^{3 \Gamma t }  S_{\rho(t)}^{L^2} $ and 
\begin{equation}
\label{eq:tego}
\begin{gathered}
\|A(t)-\Op(a(t))\|_{\operatorname{\mathscr{L}_2}}\leq C_N  e^{3 N\Gamma t }  h^{N(2- 3 \rho)} . 
\end{gathered}
\end{equation}
\end{theo}
\begin{proof}
We define 
\[ U_0 ( t ) b := (\exp t H_p)^* b , \ \ \  \partial_t U_0 ( t ) = H_p U_0 ( t ) , 
\ \ \ U_0 (0) = I ,\]
and note that using the definition \eqref{eq:rhot} and Lemma \ref{l:G2phi} we have 
\begin{equation}
\label{eq:rhots} U ( t -s ) : S_{ \rho(s) }^{L^2} \to S_{ \rho(t ) }^{L^2} ,
\end{equation}
To construct $ a ( t ) $ we start with  $a_0(t):= U_0(t)a_0$ so that
$a_0(t)\in   S_{\rho(t)}^{L^2} $.  Set $A_0(t):=\Op(a_0(t))$. Then, using Lemma~\ref{l:compose} we obtain
$$
\begin{aligned}
\dot A_0(t)&= \Op(\dot a_0(t))
= \Op( H_p a_0(t))
= \mathcal{L}_0 A_{{0}}(t)+\Op(e_0(t)), \ \ \ e_0(  t ) \in h^{(2-3 \rho) } e^{ 3  \Gamma t } S_{\rho(t)}^{L^2} . 
\end{aligned} $$

Suppose now that we found 
\[ a_j(t)\in h^{(2-3 \rho)j } e^{ 3 j \Gamma t } S_{\rho(t)}^{L^2} \ \ \ \ j=0,\dots, N-1 \]
 such that, with $A_{N-1}:=\sum_{j=0}^{N-1}\Op(a_j(t))$, we have
\begin{equation*}
\dot A_{N-1}= \mathcal{L}_0 A_{N-1}(t)+\Op(e_N(t)), \ \ \ e_N(t)\in h^{(2-3 \rho)N } e^{ 3 N \Gamma t } S_{\rho(t)}^{L^2} .
\end{equation*}
Using $ e_N $ we define
$$
a_N(t):= -\int_0^t U_0 (t-s)e_N(s)ds, \ \ \ \partial_t a_N = H_p a_N - e_N ,  \ \ a_N (0 ) = 0. 
$$
Then, using \eqref{eq:rhots}, 
$$
a_N(t)\in h^{(2-3 \rho)N } e^{ 3 N \Gamma t } S_{\rho(t)}^{L^2}, $$
and hence, with $A_N(t)=A_{N-1}(t)+\Op(a_N(t))$, we have
$$
\begin{aligned}
\dot A_N(t)&=\mathcal{L}_0 A_{N-1}(t)+\Op(e_N(t))+\Op(\dot a_N(t))\\
&=\mathcal{L}_0A_{N-1}(t)+\Op(H_p a_N(t))\\
&= \mathcal{L}_0 A_N(t)+\Op(e_{N+1}(t)), \ \ \ e_{N+1} ( t ) \in h^{(2-3 \rho)(N +1)} e^{ 3 (N+1) \Gamma t } S_{\rho(t)}^{L^2} .
\end{aligned}
$$  Note that in the last line we used Lemma~\ref{l:compose} to obtain the estimates on $e_{N+1}$. This gives $ a = \sum_{j \leq N} a_j $.

To compare $ A_N ( t) := \Op ( a ( t )) $ to $ A ( t ) $, we use the fact that $ e^{ t \mathcal L_0 } $ preserves the Hilbert--Schmidt norm
(see \eqref{eq:Schr}): 
\begin{align*}
 \big\| A(t)-A_N(t)\big\|_{\operatorname{\mathscr{L}_2}} & \leq \int_0^t \Big\|e^{(t-s)\mathcal{L}_0}\Op(e_{N+1}(s))\Big\|_{\operatorname{\mathscr{L}_2}}ds
 \leq h^{(2-3 \rho)(N +1)} e^{ 3 (N+1) \Gamma t }  .
\end{align*}
This completes the proof
\end{proof}

\section{The semigroup generated by the Lindbladian.}
\label{s:semiL}

We prove here that the Lindblad evolution is well defined in the space of Hilbert--Schmidt operators. This is done under the assumption \eqref{eq:assa} alone. 

To describe the action of $ \mathcal L $ on operators $ \mathscr S \to \mathscr S' $, we identify such operators  with their Schwartz kernels in  $ 
\mathbb R^{n} \times \mathbb R^n $  and consider 
\begin{equation} 
\label{eq:actS} 
 \mathcal L_1 : \mathscr S' ( \mathbb R^{n} \times \mathbb R^n ) \to \mathscr S' 
 ( \mathbb R^{n} \times \mathbb R^n ), \ \ \ 
 \mathcal L_0  : \mathscr S ( \mathbb R^{n} \times \mathbb R^n ) \to \mathscr S
 ( \mathbb R^{n} \times \mathbb R^n ). 
 \end{equation}
More precisely,  for $K \in\mathscr{S}'(\mathbb{R}^n\times \mathbb{R}^n)$, and 
$ \chi \in \mathscr S ( \mathbb R^n  \times \mathbb R^n ) $ we denote
by $ K ( \chi ) $ the distributional pairing, formally equal  to $ \int K ( x, y) \chi ( x,y ) dx dy $.
Then for $ A , B : \mathscr S ( \mathbb R^n ) \to \mathscr S ( \mathbb R^n ) $ we define, 
$ ( A \otimes B ) K \in \mathscr{S}'(\mathbb{R}^n\times \mathbb{R}^n) $ by
\[ (  A \otimes B ) K ( \varphi \otimes \psi ) := K ( A^t \varphi \otimes B \psi ), \ \ 
\varphi, \psi \in \mathscr S ( \mathbb R^n ) , \ \ (\varphi \otimes \psi) ( x, y ) := \varphi ( x ) \psi ( y) ,  \]
where $ A^t $ is the transpose of $ A $:   for $ f, g \in \mathscr S ( \mathbb R^n )$, 
$ (A f) ( g ) = f ( A^t g ) $ (this also defines the action of $ A $ on $ \mathscr S' $). We note that if we identify the Schwartz kernels with operators then $ ( A \times B ) K = A K B $.

In this notation 
\begin{equation}
\label{eq:defL1}   \begin{split}
 \mathcal{L}_1&:= \frac{i}{h}\big(\Op(p)\otimes I-I\otimes \Op(p)\big) 
+ \frac{\gamma}{h}\sum_j \big((\Op(\ell_j)\otimes I)(I\otimes \Op(\bar{\ell}_j) )\big) \\
& \ \ \ \ \ \ \ \ -\tfrac{1}{2}\sum_j  \big(\Op(\bar{\ell}_j)\Op(\ell_j)\otimes I +I\otimes \Op(\bar{\ell}_j)\Op(\ell_j) \big) ,
 \end{split} \end{equation}
and $\mathcal{L}_0:=\mathcal{L}_1|_{\mathscr{S} ( \mathbb R^{2n} ) }.$
 
 The following lemma describes $ \mathcal L_1 $ in a way that allows an application of
 Proposition \ref{p:spec}, which in turn provides the definition of $ \mathcal L $ as an unbounded operator on $ \mathscr L_2 $. 
\begin{lemm}
 \label{l:applyingSpec}
 The operator $ \mathcal L_1 : \mathscr S' ( \mathbb R^{2n} ) \to \mathscr S' ( \mathbb R^{2n} ) $
 defined by \eqref{eq:defL1} is given by $ \mathcal L_1 = \Op ( L ) $, where
where $L =L (x,\xi,y,\eta)\in C^\infty(\mathbb{R}^{4n})$ satisfies
 \begin{equation}
 \label{eq:estL}
 |\partial^\alpha L |\leq C_\alpha(1+|x|+|\xi|+|y|+|\eta|),\qquad |\alpha|\geq 1.
 \end{equation}
 Moreover,  identifying the Hilbert--Schmidt class $ \mathscr L_2 ( L^2 ( \mathbb R^n ) ) $ with
 $ L^2 ( \mathbb R^{n} \times \mathbb R^n ) $ using Schwartz kernels, the Lindbladian 
 $ \mathcal L $ with the domain
\begin{equation}
\label{eq:domain} 
   \mathcal D ( \mathcal L ) := \{ A \in \mathscr L_2 ( L^2 ( \mathbb R^n ) ) :
   \mathcal L_1 A \in \mathscr L_2 ( L^2 ( \mathbb R^n ) ) \} ,
 \end{equation}
  satisfies
\[  \mathcal L = \overline{ \mathcal L_0}  , \ \ \ \mathcal L^* = \overline {\mathcal L_0^* } , \]
where $ \mathcal L_0^* : \mathscr S ( \mathbb R^{n} \times \mathbb R^n ) \to \mathscr S
 ( \mathbb R^{n} \times \mathbb R^n )  $
is the formal adjoint of $ \mathcal L_0 $. 
\end{lemm}
\begin{proof}
Using coordinates $((x,\xi),(y,\eta))\in \mathbb{R}^{2n}\times \mathbb{R}^{2n}$ and 
denoting  $\Op_{\mathbb{R}^{2n}}$ 
the Weyl quantization on $ \mathbb R^{2n} $, 
the definitions above show that 
$$
\Op(a)\otimes I = \Op_{\mathbb{R}^{2n}}(a(x,\xi)),\qquad I \otimes \Op(a)= \Op_{\mathbb{R}^{2n}}(\tilde a(y,\eta)),  \ \ \ \tilde a  := e^{ i \langle h D_y, D_\eta \rangle } a  .
$$
(See \cite[Theorem 4.13]{z12}: if $ \Op (a) = \Op_1 (a_1 ) $, then $ \Op ( \tilde a ) =  \Op ( a )^t = \Op_0 ( 
a_1 ) $.) Consequently $a\in C^\infty(\mathbb{R}^{2n})$ satisfies, 
$
|\partial_{(x,\xi)}^\alpha a(x,\xi)|\leq C_{\alpha}(1+|x|+|\xi|)$ for $ |\alpha|\geq 1
$,
and,  by \cite[Theorem 4.17]{z12}, so does $ \tilde a $. 

Since $\ell_j\in S_{(1)}$, by Proposition~\ref{p:comp} we have 
$
\Op(\bar{\ell}_j)\Op(\ell_j)= \Op(c_j)
$,
for $c_j$ satisfying
$
|\partial_{(x,\xi)}^\alpha c_j(x,\xi)|\leq C_{\alpha}(1+|x|+|\xi|) $ for $  |\alpha|\geq 1
$.
Together with the facts that $\ell_j\in S_{(1)}$ and $p\in S_{(2)}$, this implies
$
\mathcal{L}_1=\Op_{ \mathbb R^{2n} } (L)$, 
where $ L $ satisfies \eqref{eq:estL}. Thus we can apply Proposition~\ref{p:spec} and the lemma follows.
\end{proof}
The next lemma describe the adjoint of $\mathcal{L}$:
\begin{lemm}
The adjoint of the Lindblad operator $\mathcal{L}$, $\mathcal{L}^*,$ is given by
\begin{equation}
\label{eq:Liad}
\mathcal{L}^*B=-\frac{i}{h}[P,B]+\frac{\gamma}{h}\sum_j L_j^*BL_j-\frac{1}{2}(L_j^*L_jB+BL_j^*L_j),
\end{equation}
with domain 
$$
\mathcal{D}(\mathcal{L}^*)=\{ A\in \mathscr{L}_2\,:\, \mathcal{L}^*A\in \mathscr{L}_2\},
$$
where for any $ A \in \mathscr L_2 $, 
$ \mathcal L^* A  $ is defined as an operator  $ \mathscr S  \to  \mathscr S'   $. 
\end{lemm}
\begin{proof}
By Proposition~\ref{p:spec} it is sufficient to compute the formal adjoint in the action on operators $ \mathscr S' \to  \mathscr S $.
Observe that, using cyclicity of the trace, for $ A, B : \mathscr S' \to  \mathscr S $, 
\begin{align*}
\Big\langle \tfrac{i}{h}[P,A],B\Big\rangle_{\mathscr{L}_2}&=\tr \Big(\tfrac{i}{h}[P,A]B^*\Big) =\tfrac{i}{h}\tr \Big((PA-AP)B^*\Big)
=\tfrac{i}{h}\tr \Big(A[P,B]^*\Big)\\
&=\tr \Big(A\big(-\tfrac{i}{h}[P,B]\big)^*\Big)=\langle A, -\tfrac{i}{h}[P,B]\rangle_{\mathscr{L}_2}, \\
\langle L_jAL_j^*,B\rangle_{\mathscr{L}_2}&=\tr \Big(L_jAL_j^*B^*\Big)=\tr \Big(AL_j^*B^*L_j\Big)
=\langle A, L_j^*BL_j\rangle_{\mathscr{L}_2},\\
\langle L_j^*L_jA,B\rangle_{\mathscr{L}_2}&=\tr \Big(L_j^*L_jAB^*\Big)=\tr \Big(AB^*L^*_jL_j\Big)
=\langle A, L_j^*L_jB\rangle_{\mathscr{L}_2},
\end{align*}
and similarly for $ \langle AL_j^*L_j,B\rangle_{\mathscr{L}_2}$.
\end{proof}

We next record some properties of $\mathcal{L}$ and its adjoint.
\begin{lemm}
\label{l:formalComputation}
For $A : \mathscr S' \to \mathscr S$, 
\begin{equation}
\label{e:real1}
2\Re \langle \mathcal{L}A,A\rangle _{\mathscr{L}_2}=-\frac{\gamma}{h}\sum_j \|[L_j,A]\|_{\mathscr{L}_2}^2+\frac{\gamma}{h}\langle \sum_j [L_j,L_j^*]A^*,A^*\rangle_{\mathscr{L}_2},
\end{equation}
and 
\begin{equation}
\label{e:real2}
2\Re \langle \mathcal{L}^*A,A\rangle_{\mathscr{L}_2} =-\frac{\gamma}{h}\sum_j \|[L^*_j,A]\|_{\mathscr{L}_2}^2+\frac{\gamma}{h}\langle \sum_j [L_j,L_j^*]A,A\rangle_{\mathscr{L}_2}.
\end{equation}
\end{lemm}
\begin{proof}
First, observe that \eqref{e:Lindblad} and \eqref{eq:Liad} show 
$$
(\mathcal{L}A)^*=\mathcal{L}A^*,\qquad (\mathcal{L}^*A)^*=\mathcal{L}^*A^*.
$$
Thus, we compute
\begin{align*}
2\Re \langle \mathcal{L}A,A\rangle_{\mathscr{L}_2} & = \tr ((\mathcal{L}A )A^*+A (\mathcal{L}A^*))\\
&=\tr\Big( \frac{i}{h}[P,A]A^* +A\frac{i}{h}[P,A^*] +\frac{\gamma}{h}\sum_j \Big(L_jAL_j^*A^* +AL_jA^*L_j^*\Big)\\
&\qquad -\frac{\gamma}{2h}\sum_j(L_j^*L_jAA^*+AL_j^*L_jA^*+AL_j^*L_jA^*+AA^*L_j^*L_j)\Big)\\
&=\tr\Big( \frac{i}{h}[P,A]A^* +A\frac{i}{h}[P,A^*] \Big)\\
&\qquad+\frac{\gamma}{h}\sum_j \tr\Big(L_jAL_j^*A^* +AL_jA^*L_j^*-L_j^*L_jAA^*-AL_j^*L_jA^*\Big).
\end{align*}
Now,
\begin{align*}
\tr\Big( [P,A]A^* +A[P,A^*] \Big)=\tr\Big( PAA^* -AA^*P \Big)=0,
\end{align*} 
and
\begin{align*}
&\tr\Big(L_jAL_j^*A^* +AL_jA^*L_j^*-L_j^*L_jAA^*-AL_j^*L_jA^*\Big)\\
&=\tr\Big(-[L_j,A]([L_j,A])^* +L_jAA^*L_j^*+AL_jL_j^*A^*-L_j^*L_jAA^*-AL_j^*L_jA^*)\\
&=\tr\Big(-[L_j,A]([L_j,A])^* +[L_j,L_j^*]A^*A).
\end{align*}
Hence,~\eqref{e:real1} follows.

The computation for~\eqref{e:real2} is similar. Since the commutator part of $\mathcal{L}^*$ has the same form as that of $\mathcal{L}$, we only need to compute
\begin{align*}
&\tr\Big(L_j^*AL_jA^* +AL_j^*A^*L_j-L_j^*L_jAA^*-AL_j^*L_jA^*\Big)\\
&=\tr \Big( -[L_j^*,A]([L_j^*,A])^* +L_j^*AA^*L_j-L_j^*L_jAA^*\Big)\\
&=\tr \Big( -[L_j^*,A]([L_j^*,A])^* +[L_j,L_j^*]AA^*\Big),
\end{align*}
and~\eqref{e:real2} follows.
\end{proof}

The next lemma will be used to control the second terms on the right hand sides of \eqref{e:real1} and \eqref{e:real2}. 
\begin{lemm}
\label{l:semibounded}
Let $C_0\in \mathbb{R}$ and suppose that $E: \mathscr S  \to L^2  $ is a self-adjoint operator on $ L^2 ( \mathbb R^n) $ satisfying $E\leq C_0$. Then,
for $ B : \mathscr S' \to \mathscr S $, 
\begin{equation}
\label{eq:posi}
\Big\langle EB,B\Big\rangle_{\mathscr{L}_2}\leq C_0\|B\|_{\mathscr{L}_2}^2.
\end{equation}
\end{lemm}
\begin{proof}
To see this, observe that exists an $L^2$-orthonormal basis $u_j $ and $\lambda_j\geq 0$,
$$
BB^*=\sum_j \lambda_j   u_j \otimes u_j ,  \ \ \   (f \otimes g) ( \varphi) := f \langle \varphi, g \rangle.  
$$
We also note that if $ \lambda_j > 0 $ then $ u_j \in \mathscr S \subset \mathcal D ( E ) $. 
Then,  
$$
\begin{aligned}
\Big\langle EB,B\Big\rangle_{\mathscr{L}_2}&=\tr (EBB^*)
=\sum_j \langle EBB^*u_j,u_j\rangle_{L^2} = \sum_{j} \lambda_j\langle Eu_j,u_j\rangle_{L^2}\\ & \leq C_0\sum_j \lambda_j =C_0 \|B\|_{\mathscr{L}_2}^2,
\end{aligned}
$$
which is \eqref{eq:posi}.
\end{proof}
Next, we provide an estimate 
\begin{lemm}
\label{l:smoothEstimate}
Suppose that, as a bounded self-adjoint operator on $L^2(\mathbb{R}^n)$ (see \eqref{eq:fric})
\begin{equation}
\label{eq:sumLL}
\sum_j [L_j,L_j^*]\leq \frac{2M h}{\gamma}.
\end{equation}
Then, for $A : \mathscr S' \to \mathscr S $ and $\lambda>0$,
\begin{equation}
\label{e:injective1}
\lambda \|A\|_{\mathscr{L}_2}\leq \|(\mathcal{L}-M-\lambda)A\|_{\mathscr{L}_2}, \ \ \
\lambda \|A\|_{\mathscr{L}_2}\leq \|(\mathcal{L}^*-M-\lambda)A\|_{\mathscr{L}_2}, 
\end{equation}
\end{lemm}
\begin{proof}
Observe that by Lemma~\ref{l:formalComputation}, and Lemma~\ref{l:semibounded}
\begin{align*}
2\Re \langle (\mathcal{L}-M-\lambda)A,A\rangle _{\mathscr{L}_2}&\leq -2\lambda \|A\|_{\mathscr{L}_2} -2M\|A\|_{\mathscr{L}_2}+\frac{\gamma}{h} \Big\langle\sum_j [L_j,L_j^*]A^*,A^*\Big\rangle_{\mathscr{L}_2}\\
&\leq -2\lambda \|A\|_{\mathscr{L}_2}.
\end{align*}
Hence,
$$
2\lambda \|A\|_{\mathscr{L}_2}\leq |2\Re \langle (\mathcal{L}-M-\lambda)A,A\rangle _{\mathscr{L}_2}|\leq 2\|(\mathcal{L}-M-\lambda)A\|_{\mathscr{L}_2}\|A\| _{\mathscr{L}_2},
$$
from which the first estimate in~\eqref{e:injective1} follows.
The argument for the second estimate is identical.
\end{proof}

\begin{prop}
\label{p:HY}
Suppose that \eqref{eq:sumLL} holds.
Then the operator $\mathcal{L}$ with domain 
$
\mathcal{D}(\mathcal{L}):=\{ A\in \mathscr{L}_2\,:\, \mathcal{L}A\in \mathscr{L}_2\}
$
generates a strongly continuous semigroup 
$$ e^{t\mathcal{L}}:\mathscr{L}_2\to \mathscr{L}_2 \ \ \text{ and } \ \
\|e^{t\mathcal{L}}\|_{\mathscr{L}_2\to \mathscr{L}_2}\leq e^{Mt}, 
\ \ \ \  t \geq 0. 
$$
\end{prop}
\begin{proof}
By Proposition~\ref{p:spec}, or rather its proof (see \eqref{eq:closure1}), and Lemma~\ref{l:applyingSpec},  for $A\in\mathcal{D}(\mathcal{L})$ there exists a sequence of operators $A_n : \mathscr S' \to \mathscr S $ such that $A_n\overset{\mathscr{L}_2}{\longrightarrow} A$ and $\mathcal{L}A_n\overset{\mathscr{L}_2}{\longrightarrow} \mathcal{L}A$. Hence, for $A\in\mathcal{D}(\mathcal{L})$ and $\lambda>0$, Lemma~\ref{l:smoothEstimate} gives
\begin{align*}
\lambda \|A\|_{\mathscr{L}_2}&=\lambda \lim_{n\to \infty}\|A_n\|_{\mathscr{L}_2}\\
&\leq  \lim_{n\to \infty }\|(\mathcal{L}-M-\lambda)A_n\|_{\mathscr{L}_2}=\|(\mathcal{L}-M-\lambda)A\|_{\mathscr{L}_2}.
\end{align*}

Similarly, for $A\in\mathcal{D}(\mathcal{L}^*)$, we have $A_n:\mathscr{S}'\to \mathscr{S}$ such that $A_n\overset{\mathscr{L}_2}{\longrightarrow} A$ and $\mathcal{L}^*A_n\overset{\mathscr{L}_2}{\longrightarrow} \mathcal{L}^*A$. This implies that for $A\in\mathcal{D}(\mathcal{L}^*)$, and $\lambda>0$
\begin{align*}
\lambda \|A\|_{\mathscr{L}_2}\leq \|(\mathcal{L}^*-M-\lambda)A\|_{\mathscr{L}_2}.
\end{align*}
In particular, $(\mathcal{L}-M-\lambda)^{-1}:\mathscr{L}_2\to \mathcal{D}(\mathcal{L})$ exists and satisfies,
$$
\|(\mathcal{L}-M-\lambda)^{-1}\|_{\mathscr{L}_2\to \mathscr{L}_2}\leq \lambda^{-1}, 
 \ \ \ \lambda > 0 .
$$
The Hille-Yosida theorem then implies that $\mathcal{L}-M$ generates a strongly continuous semigroup $e^{t(\mathcal{L}-M)}$ satisfying
$$
\|e^{t(\mathcal{L}-M)}\|_{\mathscr{L}_2\to \mathscr{L}_2}\leq 1,
$$
from which the proposition follows.
\end{proof}

We conclude this section by showing how condition \eqref{eq:sumLL} is related to a lower bound on the friction \eqref{eq:fric}
\begin{lemm}
\label{cl2qu}
Let 
\begin{equation}
\label{eq:sumbr}
M_0 := \sup_{ \mathbb R^{2n}} \mu, \ \ \ \mu := \frac 1 {2i }\sum_{j=1}^J  \{ \ell_j, \bar \ell_j \}. 
\end{equation}
Then there is $C_0>0$ such that \eqref{eq:sumLL} holds with 
\begin{equation}
\label{eq:MM}  M =  \gamma M_0 + C_0 h \gamma , 
\end{equation}
for  $0 < h < 1 $.  Furthermore, if $\mu\equiv 0$, then~\eqref{eq:sumLL} holds with 
\begin{equation}
\label{e:MM2}
M=C_0h^2\gamma
\end{equation} for $0<h<1$.
\end{lemm}
\begin{proof}
Since \eqref{eq:assa} shows that $ \mu \in S(1) $, the first estimate is a straightforward application of sharp G{\aa}rding inequality for the class $ S(1) $ -- see 
\cite[Theorem 7.1]{DiSj} or \cite[\S 4.7.2]{z12}. When $\mu\equiv0$, we use that $[L_j,L_j^*]=\Op(\frac{h}{2i}\{\ell_j,\bar{\ell}_j\}+h^3 e)$ for some $e\in S(1)$ and hence the second estimate follows.
\end{proof}

\section{The Classical Dynamics}
\label{s:cla}

It will be convenient to rewrite the Lindbladian as 
$$
\mathcal{L}A= \frac{i}{h}[P,A]+\frac{\gamma}{2h}\sum_j\left( [L_jA,L_j^*]+[L_j,AL_j^*] \right) .
$$
Our first goal is to motivate the classical Fokker--Planck equation \eqref{eq:FP} from the
evolution equation for $ \mathcal L $.

Observe that for $0\leq \rho <1$, and $a\in S^{L^2}_{\rho}$,
$$
L_j A=\Op(\ell_j a+\frac{h}{2i}\{\ell_j,a\}+ h^{2-2\rho}e_1),\qquad AL_j^*=\Op(a\bar{\ell}_j+\frac{h}{2i}\{a,\bar{\ell_j}\}+h^{2-2\rho}e_2),
$$
with $e_j\in S^{L^2}_{\rho}$.
Hence for $a\in S^{L^2}_{\rho}$
\begin{equation}
\label{eq:main}
\begin{split}
\mathcal{L}A  = & \Op(H_pa) +\frac{\gamma}{2i} \sum_j \Op( (2 \{\ell_j,\bar{\ell}_j\}a -\ell_j H_{\bar{\ell}_j}a+\bar{\ell}H_{\ell}a) )\\
&  +\frac{h\gamma}{4}\sum_j \Op(H_{\bar{\ell}_j}H_{\ell_j} a +H_{\ell_j}H_{\bar{\ell_j}}a)+ h^{2 -3 \rho}(1+\gamma)\Op(e),
\end{split}
\end{equation}
with $e\in S^{L^2}_{\rho}$. {Heuristic arguments in the physics literature -- see \cite{hrr} and the discussion and references given there -- suggest that the natural classical evolution should be given by the equation up to the diffusion term
$ \sum_j H_{\bar{\ell}_j}H_{\ell_j} +H_{\ell_j}H_{\bar{\ell_j}}$ which is a non-positive differential operator acting on the classical observable $ a $ (see \eqref{eq:exam} for a striking example). 
Hence as the generator of the classical flow (a form of Fokker--Planck operator) we take
 $Q\in \operatorname{Diff}^2(\mathbb{R}^{2n})$ given by 
$$
Q:=H_p+\frac{\gamma}{2i}\sum_j (2 \{\ell_j,\bar{\ell}_j\} -\ell_j H_{\bar{\ell}_j}+\bar{\ell}_jH_{\ell_j}) +\frac{h\gamma}{4}\sum_j(H_{\bar{\ell}_j}H_{\ell_j}  +H_{\ell_j}H_{\bar{\ell}_j}).
$$}

The key estimate for evolution by $Q$ is given as follows. We need here the additional technical assumption \eqref{eq:tech}.  To state the next estimate we recall
the definition of semiclassical Sobolev norms:
\begin{equation}\label{eq:Sob}
\| u \|_{ H_\varepsilon ^s }^2 := \int ( 1 + | \varepsilon \zeta |^2 )^s | \widehat u ( \zeta ) |^2 d \zeta, 
 \ \ \  \widehat u ( \zeta ) := \int u ( z ) e^{- i z \zeta } d z . \end{equation}
\begin{prop}
\label{p:advect}
Suppose that \eqref{eq:tech} holds, 
and $\gamma \leq h^{-1}$.
Let $U(t):L^2(\mathbb{R}^{2n})\to L^2(\mathbb{R}^{2n})$ be defined by 
\begin{equation}
\label{e:defU}
(\partial_t -Q )U(t)=0,\qquad U(0)=\Id.
\end{equation}
Then, for all $s\geq 0$, there is $C>0$ such that for all $t\geq 0$,
\begin{equation}
\label{eq:estUg}
\|U(t)  \|_{H_{\varepsilon}^s\to H_{\varepsilon}^s}\leq Ce^{M_0\gamma t}  \sqrt{1+t\gamma \varepsilon}  ,\qquad \|U(t)\|_{L^2\to L^2}\leq Ce^{M_0\gamma t}.
\end{equation}
where $ M_0 $ is given in \eqref{eq:sumbr}, the norms are defined in \eqref{eq:Sob}, and
\begin{equation}
\label{eq:vareps}
\varepsilon := \sqrt{ \gamma h } .
\end{equation}
If, $\sum_{j}\{\bar{\ell}_j,\ell_j\} \equiv 0$, that is there is no friction \eqref{eq:fric},  then 
\begin{equation}
\label{eq:estU0}
\|U(t)  \|_{H_{\varepsilon}^s\to H_{\varepsilon}^s}\leq C.
\end{equation}
\end{prop}

\noindent
{\bf Remark.}  The estimates \eqref{eq:estUg} and \eqref{eq:estU0} do not address the smoothing effect of the evolution by 
\eqref{e:defU}. Obtaining quantitative estimates seems to require stronger assumptions than \eqref{eq:assa} and we restrict ourselves
to that case. 

\begin{proof}
Recall from \eqref{eq:FP} and \eqref{eq:Bj} that $ Q $ is given by 
$ H_p+ \gamma \sum_j B_j + \mu $ plus a second order divergence form operator and the first two terms are anti-selfadjoint. Hence, for $u\in H^2$, 
$$
\Re \langle Qu, u\rangle= \frac{\gamma}{2i}\sum_j \langle \{ \ell_j,\bar{\ell}_j\}u,u\rangle -\frac{h\gamma}{4}\sum_j\left(\|H_{\ell_j}u\|_{L^2}^2+\|H_{\bar{\ell}_j}u\|_{L^2}^2\right).
$$
We start with  an estimate on the solution, $v,$ to 
\begin{equation}
\label{eq:conM0}
e^{ -t M_0 \gamma}  (\partial_t-Q ) (e^{ t M_0 \gamma} v (t)  ) = 
 (\partial_t-Q+M_0\gamma ) v(t) =f,\qquad v(0)=v_0.
\end{equation}
We have
$$
\begin{aligned}
\langle f,v\rangle&=\Re \langle (\partial_t- Q+M_0\gamma)v,v\rangle\\ 
&=\frac{1}{2}\partial_t\|v\|_{L^2}^2+\gamma \langle (M_0- \sum_j\tfrac{1}{2i}\{ \ell_j,\bar{\ell}_j\})v,v\rangle +\frac{h\gamma}{4}\sum_j(\|H_{\ell_j}v\|_{L^2}^2+\|H_{\bar{\ell}_j}v\|_{L^2}^2).
\end{aligned}
$$
Hence, 
$$
\partial_t\|v\|_{L^2}^2 +\frac{h\gamma}{2}\sum_j\left(\|H_{\ell_j}u\|_{L^2}^2+\|H_{\bar{\ell}_j}u\|_{L^2}^2\right)\leq 
2|\langle f,v\rangle|.
$$
For $T>0$ the ellipticity hypothesis \eqref{eq:nondeg} then gives
\begin{equation}
\label{e:basicEnergy}
\begin{aligned}
\|v(T)\|_{L^2}^2+\gamma h c\int_0^T  \|\nabla v\|_{L^2}^2
& \leq \|v(T)\|_{L^2}^2+\frac{h\gamma}{2}\int_0^T\sum_j\left(\|H_{\ell_j}v\|_{L^2}^2+\|H_{\bar{\ell}_j}v\|_{L^2}^2\right) \\
&\leq 2\int_0^T\|f(t)\|_{L^2}\|v(t) \|_{L^2}dt+\|v_0\|_{L^2}^2.
\end{aligned}
\end{equation}

Now let $u$ solve 
$$
(\partial_t- Q+M_0\gamma )u=0,\qquad u(0)=u_0.
$$
Then, applying~\eqref{e:basicEnergy}, we obtain
\begin{equation}
\label{eq:kenergy}
 \|u(T)\|_{L^2}^2+ c\int_0^T\| \varepsilon \nabla u (t) \|_{L^2}^2ds\leq\|u_0\|_{L^2}^2 , 
  \ \ \ \varepsilon = \sqrt{ 
  \gamma h}. 
\end{equation}
To proceed by induction let us assume that for $k\geq 0$
\begin{equation}
\label{e:induction}
\begin{aligned}
&\sum_{ 0\leq |\alpha|\leq k}\|(\varepsilon\partial)^{\alpha} u(T)\|_{L^2}^2+ \int_0^T \sum_{{1} \leq |\alpha|\leq k+1}\|(\varepsilon\partial)^{\alpha} u (t) \|_{L^2}^2\\
&\qquad \leq C \sum_{|\beta|\leq k}\|(\varepsilon\partial)^\beta u_0\|_{L^2}^2+ {C }T\gamma \varepsilon \|u_0\|^2_{L^2}.
\end{aligned}
\end{equation}
We  set 
\begin{gather*}
Q_1:=\sum_j i \{\bar{\ell}_j,\ell_j\},\quad Q_2:=\frac{1}{2i}\sum_j( -\ell_j H_{\bar{\ell}_j}+\bar{\ell}_jH_{\ell_j}),\\
 Q_3:=\frac{1}{4}\sum_j(H_{\bar{\ell}_j}H_{\ell_j}  +H_{\ell_j}H_{\bar{\ell}_j}).
\end{gather*}
so that
\begin{equation}
\label{eq:Qal}
\begin{split} 
 (\partial_t-  Q + {M_0} )\partial^\alpha u & =[H_p,\partial^\alpha ]u+\gamma [Q_1,\partial^\alpha]u+\gamma [Q_2,\partial^\alpha]u+\gamma h [Q_3,\partial^\alpha]u, \\
 \partial^\alpha (0)  & =\partial ^\alpha u_0.
\end{split} 
\end{equation}
We have the following estimates on the commutators appearing on the right hand side: 
\begin{equation}
\label{eq:estcom}
\begin{gathered}
\|[H_p,\partial^\alpha ]u\|_{L^2}\leq C\sum_{1\leq |\beta| \leq |\alpha|}\|\partial^\beta u\|_{L^2},\qquad \|[Q_1,\partial^\alpha ]u\|_{L^2}\leq C\sum_{0\leq |\beta|\leq |\alpha|-1}\| \partial^\beta u\|_{L^2},\\
\|[Q_2,\partial^\alpha ]u\|_{L^2}\leq C\sum_{1\leq |\beta|\leq |\alpha|}\|\partial^\beta u\|_{L^2},\qquad \|[Q_3,\partial^\alpha ]u\|_{L^2}\leq C\sum_{1\leq |\beta|\leq |\alpha|+1}\|\partial^\beta u\|_{L^2}. 
\end{gathered}
\end{equation}
It is important here that in the estimates {\em not} involving $ Q_1 $, we have $ |\beta | \geq 1 $ on the right hand sides. To obtain the estimate on commutators with $Q_2$, we use assumption~\eqref{eq:tech}.

Applying ~\eqref{e:basicEnergy} to \eqref{eq:Qal} and using \eqref{eq:estcom} we obtain
\begin{align*}
&\sum_{|\alpha|=k+1}\|\partial^\alpha u(T)\|_{L^2}^2+c\gamma h \int_0^T \sum_{|\alpha|=k+1} \|\nabla \partial^\alpha u\|_{L^2}^2dt \\
& \ \ \  \leq C\int_0^T\sum_{|\alpha|=k+1}\Big(\sum_{1\leq |\beta | \leq k+1}\|\partial^\beta u\|_{L^2}+ \gamma\|u\|_{L^2}+\gamma h \sum_{1\leq |\beta'|\leq k+2}\|\partial^{\beta'} u\|_{L^2}\Big)\|\partial^\alpha u\|_{L^2}dt\\
& \ \ \ \ \ \ \ \ \ \ \ \ \ \ \ \ +\sum_{|\alpha|=k+1}\|\partial^\alpha u_0\|_{L^2}^2\\
& \ \ \ \leq C\int_0^T\sum_{|\alpha|=k+1}\Big(\sum_{1\leq |\beta| \leq k+1}\|\partial^\beta u\|_{L^2}+\gamma h\sum_{ |\beta'|= k+2}\|\partial^{\beta'} u\|_{L^2}\Big)\|\partial^\alpha u\|_{L^2}dt\\
& \ \ \ \ \ \ \ \ \ \ \ \ \ \ \ \   +\sum_{|\alpha|=k+1}\|\partial^\alpha u_0\|_{L^2}^2 +CT\gamma \|u_0\|^2_{L^2}.
\end{align*}
 Young's inequality ($ 2 ab \leq \delta^{-1} a^2 + \delta b^2 $) allows us to move the highest order terms 
 from the right hand side to the left hand side and that gives 
\begin{align*}
&\sum_{|\alpha|=k+1}\|\partial^\alpha u(T)\|_{L^2}^2+c\gamma h \int_0^T \sum_{|\alpha|=k+1} \|\nabla \partial^\alpha u\|_{L^2}^2dt \\
&\leq C\int_0^T\sum_{1\leq |\beta| \leq k+1}\|\partial^\beta u\|^2_{L^2}dt+\sum_{|\alpha|=k+1}\|\partial^\alpha u_0\|_{L^2}^2 +CT\gamma \|u_0\|^2_{L^2}.
\end{align*}
We now use the inductive hypothesis \eqref{e:induction} (with $ \varepsilon = \sqrt{ \gamma h } \leq 1 $) to obtain
\begin{align*}
&\sum_{|\alpha|=k+1}\|(\varepsilon \partial)^\alpha u(T)\|_{L^2}^2+c \int_0^T \sum_{|\beta|=k+2} \|(\varepsilon  \partial)^\beta u\|_{L^2}^2dt \\
&\ \ \ \leq C\sum_{1\leq |\beta| \leq k+1} \varepsilon^{k+1-|\beta|}\int_0^T \|(\varepsilon \partial)^{\beta} u\|^2_{L^2}dt +\sum_{|\alpha|=k+1}\|(\varepsilon\partial)^\alpha u_0\|_{L^2}^2\\
&\ \ \ \qquad+CT\gamma \varepsilon^{k+1}\|u_0\|^2_{L^2}.\\
&\ \ \ \leq C\Big(\sum_{1\leq| \beta |\leq k+1}\varepsilon^{k+1-|\beta|}(\sum_{|\alpha|\leq|\beta| }\|(\varepsilon \partial)^{\alpha}u_0\|_{L^2}^2+C T\gamma (\gamma h)\|u_0\|_{L^2}^2) \\
&\ \ \ \qquad +\sum_{|\alpha|=k+1}\|(\varepsilon\partial)^\alpha u_0\|_{L^2}^2+CT\gamma \varepsilon^{k+1}\|u_0\|^2_{L^2}\\
&\ \ \ \leq C\sum_{|\alpha|\leq k+1}\|(\varepsilon\partial)^\alpha u_0\|_{L^2}^2 + {C} T\gamma \varepsilon \|u_0\|^2_{L^2}.
\end{align*}
Combined with the inductive hypothesis this shows that~\eqref{e:induction} holds with $k$ replaced by $k+1$.

Returning to \eqref{eq:conM0} we see that \eqref{e:induction} gives \eqref{eq:estUg}. When 
$ \sum_j \{ \bar \ell_j, \ell_j \} \equiv 0 $ then we can take $ M_0  $ and $ Q_2 \equiv 0 $ in the proof and that
gives \eqref{eq:estU0} (note that in this case $Q_1$ vanishes and hence the last term on the right hand side of \eqref{e:induction} does not appear).
\end{proof}

\section{Agreement of quantum and classical dynamics}
\label{s:aqd}

In this section we obtain an accurate approximation to the solution of the Lindblad master equation
which is a far reaching strengthening of Theorem \ref{t:1} in \S \ref{s:intr}. 
\begin{theo}
\label{t:2}
Suppose that $ \mathcal L $ is given by \eqref{e:Lindblad}, assumptions \eqref{eq:assa}, \eqref{eq:nondeg}, and \eqref{eq:tech} hold, 
$  h^{2 \rho - 1 }\leq \gamma \leq h^{-1} $ for some $ 0 \leq \rho \leq \frac23 $. 
There is $C_0>0$ such that if $A(t)$ satisfies (in the notation of \S \ref{s:sym})
\begin{equation*}
\partial_t A ( t) = \mathcal L A ( t ) , \  \  \ A(0) = \Op (a_0 ) , \ \  a_0 \in S^{L^2}_{\rho} ,  
\end{equation*}
then, for every $ N $ there exist $C_N>0$ and $ a ( t ) \in S^{L^2}_{\rho} $ such that 
\begin{equation}
\label{eq:t2}
\begin{gathered}
\|A(t)-\Op(a(t))\|_{\operatorname{\mathscr{L}_2}}\leq C_Ne^{(M_0 + C_0 h ) \gamma t} {(1+\gamma)^{N+1}}
( 1 +  \gamma^{\frac32} h^{\frac12}  t)^{\frac{N+2}2} 
t^{N+1}h^{(2- 3 \rho)(N+1)}, \\ 
 a(t)-U(t)a_0 \in e^{M_0\gamma t}t h^{(2 -3 \rho) }(1+\gamma)
 (1+ t\gamma^{\frac 32} h^{\frac12}  )  S^{L^2}_{\rho}, 
\end{gathered}
\end{equation}
where $U(t) $ was defined by~\eqref{e:defU}. 

If $ \sum_j \{ \ell_j , \bar \ell_j \} \equiv 0 $, that is there is no friction \eqref{eq:fric}, then 
\begin{equation}
\label{eq:t20}
\begin{gathered}
\|A(t)-\Op(a(t))\|_{\operatorname{\mathscr{L}_2}}\leq C_N e^{C_0{h^2}\gamma t}
(1+\gamma)^{N+1}t^{N+1}h^{(2- 3 \rho)(N+1)}, \\ 
 a(t)-U(t)a_0 \in  h^{(2 -3 \rho) }t(1+\gamma)S^{L^2}_{\rho}, 
\end{gathered}
\end{equation}
\end{theo}

\begin{proof}[Proof of Theorem \ref{t:1Gauss} assuming Theorem \ref{t:2}] 
Let $a_0\in S^{L_2}_{1/2}$. Then  observe that by Proposition~\ref{p:HY} and Lemma~\ref{cl2qu}, together with the fact that for $a_0\in S^{L_2}_{1/2} $, $\|\Op(a_0)\|_{\mathscr{L}_2}\leq C$, 
$$
\|A(t)\|_{\mathscr{L}_2}\leq Ce^{(M_0+C_0h)\gamma t}.
$$
Next, using Proposition~\ref{p:advect}
$$
\|\Op(U(t)a_0)\|_{\mathscr{L}_2}\leq Ce^{M_0\gamma t}. 
$$ 
Therefore, since our estimates are trivially valid when $t(\gamma +\gamma^{-\frac{3}{2}})h^{\frac12}$, we may assume without loss of generality that $t(\gamma +\gamma^{-\frac{3}{2}})h^{\frac12}\leq 1$.

We now consider two cases: $\gamma=h^{2\rho-1}$ for some $\rho\geq \frac12 $ and $\rho=\frac{1}{2}$ with $\gamma \geq 1$. Observe that when $\gamma=h^{2\rho-1}$ for some $\rho\geq \frac{1}{2}$, then, using that $\gamma \leq 1$, the estimate~\eqref{eq:t2} reads
\begin{equation}
\label{eq:t3}
\begin{gathered}
\|A(t)-\Op(a(t))\|_{\operatorname{\mathscr{L}_2}}\leq C_Ne^{(M_0 + C_0 h ) \gamma t} (th^{\frac12}\gamma^{-\frac32})^{N+1} (1+ t \gamma^{\frac32} h^{\frac12}   )^{\frac{N+2}2} 
\\ 
 a(t)-U(t)a_{\lambda_h} \in e^{(M_0 + C_0 h ) \gamma t } (th^{\frac12}\gamma^{-\frac32})  (1+ t \gamma^{\frac32} h^{\frac12}   )S^{L^2}_{\rho}.
\end{gathered}
\end{equation}
Hence, since $t\gamma^{-\frac32}h^{\frac12}\leq 1$, the estimate~\eqref{eq:t21} follows in this case.
On the other hand, when $\gamma \geq 1$ and we set $\rho=\frac 12$, the estimate~\eqref{eq:t2} reads
\begin{equation}
\label{eq:t4}
\begin{gathered}
\|A(t)-\Op(a(t))\|_{\operatorname{\mathscr{L}_2}}\leq C_Ne^{(M_0 + C_0 h ) \gamma t} {t h^{\frac12} \gamma^{N+1}}t^{N+1}h^{\frac{1}{2}(N+1)} (1+ t \gamma^{\frac32} h^{\frac12}   )^{\frac{N+2}2}  ,
 \\ 
 a(t)-U(t)a_{\lambda_h} \in e^{(M_0 + C_0 h ) \gamma t}  t h^{\frac{1}{2} }\gamma  (1+ t \gamma^{\frac32} h^{\frac12}   )  S^{L^2}_{\frac{1}{2}}, 
\end{gathered}
\end{equation}
Taking $N=0$ and using $t\gamma h^{\frac12}\leq 1$, we obtain
$$
\|A(t)-\Op(U(t)a_{\lambda_h})\|_{\operatorname{\mathscr{L}_2}}\leq C e^{(M_0 + C_0 h ) \gamma t } th^{\frac{1}{2} }(1+ t \gamma^{\frac32} h^{\frac12}   )  . $$
\end{proof}

\begin{proof}[Proof of Theorem \ref{t:2}]
Define $a_0(t):= U(t)a_0$, with $U$ given in~\eqref{e:defU}. Then, recalling that $ \varepsilon =
\sqrt { \gamma h/2 } $,  $ h^\rho \leq \varepsilon \leq 1 $, 
\eqref{eq:estUg} gives 
$$a_0(t)\in e^{M_0\gamma t}(1+ t\gamma^{\frac32} h^{\frac12}  )^{\frac12}   S_{\rho}^{L^2} ,  \ \ \text{uniformly in $t\geq 0$.} $$
 Set $A_0(t):=\Op(a_0(t))$. Then, using Lemma~\ref{l:compose} as in the derivation of \eqref{eq:main}, we obtain
$$
\begin{aligned}
\dot A_0(t)&= \Op(\dot a_0(t))
= \Op(Qa_0(t))
= \mathcal{L}A_{{0}}(t)+\Op(e_1(t)),
\end{aligned}
$$
where 
$$ e_1(t)\in h^{2-3 \rho} (1+\gamma)e^{M_0\gamma t}
(1+ t\gamma^{\frac32} h^{\frac12}  )^{\frac12}   
S^{L^2}_{\rho}.$$
Suppose, by induction that we have found 
\[ a_j(t)\in e^{M_0\gamma t}t^j h^{(2 -3 \rho) j}(1+\gamma)^j
(1+ t\gamma^{\frac32} h^{\frac12}  )^{\frac{j+1}2}
S^{L^2}_{\rho}, \ \ \ \ j=0,\dots, N-1 \]
 such that, with $A_{N-1}:=\sum_{j=0}^{N-1}\Op(a_j(t))$, we have
\begin{equation*}
\dot A_{N-1}= \mathcal{L}A_{N-1}(t)+\Op(e_N(t)),
\end{equation*}
with
$$e_N(t)\in e^{M_0\gamma t} t^{N-1} 
(1+ t\gamma^{\frac32} h^{\frac12}  )^{\frac{N}2}
(1+\gamma)^N S^{L^2}_{\rho}. $$
Using $ e_N $ we define
$$
a_N(t)= -\int_0^t U (t-s)e_N(s)ds.
$$
Since, 
\[ \begin{split} \int_0^t s^{N-1} ( 1 + \gamma^{\frac32} h^{\frac12}  s )^{\frac N 2 } ( 1 + \gamma^{\frac32} h^{\frac12}  ( t - s) )^{\frac12} ds 
& \leq  ( 1 +  \gamma^{\frac32} h^{\frac12}  t)^{\frac{N+1}2}  t^N / ( N + 1 ) , 
\end{split} \]
$$
a_N(t)\in e^{M_0\gamma t}t^N h^{(2 - 3 \rho )N}  
 ( 1 +  \gamma^{\frac32} h^{\frac12}  t)^{\frac{N+1}2} 
 S^{L^2}_{\rho},
$$
and hence, with $A_N(t)=A_{N-1}(t)+\Op(a_N(t))$, we have
$$
\begin{aligned}
\dot A_N(t)&=\mathcal{L}A_{N-1}(t)+\Op(e_N(t))+\Op(\dot a_N(t))\\
&=\mathcal{L}A_{N-1}(t)+\Op(Q a_N(t))\\
&= \mathcal{L}A_N(t)+\Op(e_{N+1}(t)),
\end{aligned}
$$
with 
$$e_{N+1}\in e^{M_0\gamma t}t^{N}(1+\gamma)^{N+1}h^{(2 - 3 \rho ) (N+1)} 
( 1 +  \gamma^{\frac32} h^{\frac12}  t)^{\frac{N+1}2}
S^{L^2}_{\rho}. $$  Note that in the last line we used Lemma~\ref{l:compose} to obtain the estimates on $e_{N+1}$. This gives $ a = \sum_{j \leq N} a_j $.

We next use Proposition \ref{p:HY} and Lemma \ref{cl2qu} to compare $A(t)$ and $A_N(t)$:
\begin{align*}
 \big\| A(t)-A_N(t)\big\|_{\operatorname{\mathscr{L}_2}} & \leq \int_0^t \big\|e^{(t-s)\mathcal{L}}\Op(e_{N+1}(s))\big\|_{\operatorname{\mathscr{L}_2}}ds\\
&  \leq  C_Ne^{(M_0 + C_0 h ) \gamma t} (1+\gamma)^{N+1}t^{N+1}h^{(2- 3 \rho)(N+1)}
( 1 +  \gamma^{\frac32} h^{\frac12}  t)^{\frac{N+2}2} .
\end{align*}
The stronger version under the assumption that $ \sum_j \{ \ell_j, \bar \ell_j \} = 0 $ follows from the stronger estimates in {\eqref{e:MM2} and} \eqref{eq:estU0}.
\end{proof}

\section{Bounds on the Hilbert--Schmidt norm of Lindblad Evolution}
\label{s:lindbladBound}

In this section we use Theorem~\ref{t:2} to give lower bounds on the Hilbert--Schmidt norm of the Lindblad evolution in the case of Example \eqref{eq:exam}. We will consider two special cases: quadratic Hamiltonians and confining Hamiltonians.

\subsection{Quadratic Hamiltonians}

We first show that when $p$ is quadratic and the initial condition is Gaussian, it is possible to solve~\eqref{eq:FP} exactly. For the purposes of this section, we let $B$ be a real, symmetric, matrix and suppose that 
\begin{equation}
\label{e:quadraticP}
p(x,\xi):= \frac{1}{2}\langle B\rho,\rho\rangle,\qquad \rho:=\begin{pmatrix}x\\\xi\end{pmatrix}.
\end{equation}
We also use the notation $\Omega:=\begin{pmatrix}0&I\\-I&0\end{pmatrix}$ for the standard symplectic form.
\begin{lemm}
\label{l:quadraticP}
Let $p$ as in~\eqref{e:quadraticP}, $A_0$ be a real, symmteric, positive definite matrix. $\rho_0\in \mathbb{R}^{2n}$, and $u$ solve
\begin{equation}
\label{e:quadraticFP}
(\partial_t -H_p -\tfrac{1}{2}\gamma h\Delta_{x,\xi})u=0,\qquad u(0)=\exp\Big(-\frac{1}{2h}\langle A_0(\rho-\rho_0,\rho-\rho_0)\rangle\Big).
\end{equation}
Then,
$$
u(t)= e^{f(t)}\exp\Big(-\frac{1}{2h}\langle A(t)(\rho-\rho_0(t),\rho-\rho_0(t))\rangle\Big),
$$
where $A(0)=A_0$, $\rho_0(0)=\rho_0$, $f(0)=0$, and
\begin{equation}
\label{e:parameterEquations}
\begin{gathered}\dot \rho_0(t)=-\Omega B\rho_0\\ \dot A(t)=(A+A^t)\Omega B-\tfrac \gamma 4 (A+A^t)^2\qquad \dot f(t)=-\tfrac\gamma 4 \tr (A(t)+A^t(t)).
\end{gathered}
\end{equation}
\end{lemm}
\begin{proof}
We compute
\begin{align*}
\partial_t u&=u\Big(\dot f -\frac{1}{2h}\langle \dot A(\rho-\rho_0),\rho-\rho_0\rangle +\frac{1}{2h}\langle A\dot\rho_0,\rho-\rho_0\rangle+\frac{1}{2h}\langle A(\rho-\rho_0),\dot\rho_0\rangle\Big),\\
&=u\Big(\dot f -\frac{1}{2h}\langle \dot A(\rho-\rho_0),\rho-\rho_0\rangle +\frac{1}{2h}\langle (A+A^t)\dot\rho_0,\rho-\rho_0\rangle\Big),\\
H_p&u= u\Big(-\frac{1}{2h}\langle A\Omega B \rho, (\rho-\rho_0)\rangle -\frac{1}{2h}\langle A(\rho-\rho_0), \Omega B \rho\rangle  \Big)\\
&=u\Big(-\frac{1}{2h}\langle (A+A^t)\Omega B \rho_0, \rho-\rho_0\rangle -\frac{1}{2h}\langle (A+A^t)\Omega B(\rho-\rho_0),\rho-\rho_0\rangle\Big) \\
\gamma h \Delta u&= u\gamma h\Big(\frac{1}{4h^2}\langle (A+A^t)(\rho-\rho_0),(A+A^t)(\rho-\rho_0)\rangle -\frac{1}{2h}2\tr (A+A^t)\Big)\\
&=u\gamma h\Big(\frac{1}{4h^2}\langle (A+A^t)^2(\rho-\rho_0),\rho-\rho_0\rangle -\frac{1}{2h}\tr (A+A^t)\Big).
\end{align*}
Then, using that $u$ satisfies~\eqref{e:quadraticFP}, and equating terms by homogeneity in $\rho-\rho_0$, we obtain~\eqref{e:parameterEquations}.
\end{proof}

\noindent 
{\bf Remark.} As an easy corollary of Lemma~\ref{l:quadraticP}, we see that if $A_0=2I$, $B=0$, then 
$$
u(t)= \frac{1}{(1+2\gamma t)^{n}}e^{-\frac{1}{h(1+2\gamma t)}\langle \rho-\rho_0,\rho-\rho_0\rangle },\qquad \|u(t)\|_{L^2}=\Big(\frac{ \pi h}{2(1+2\gamma t)}\Big)^{\frac{n}{2}}.
$$
When $p=0$, the Lindblad evolution is exactly given by the Fokker--Planck evolution, and thus the solution $A(t)$ to~\eqref{eq:evolGauss}, satisfies
$$
\|A(t)\|_{\mathscr{L}_2}= \Big(\frac{ 1}{1+2\gamma t}\Big)^{\frac{n}{2}}.
$$

\begin{figure}
\includegraphics[width=12cm]{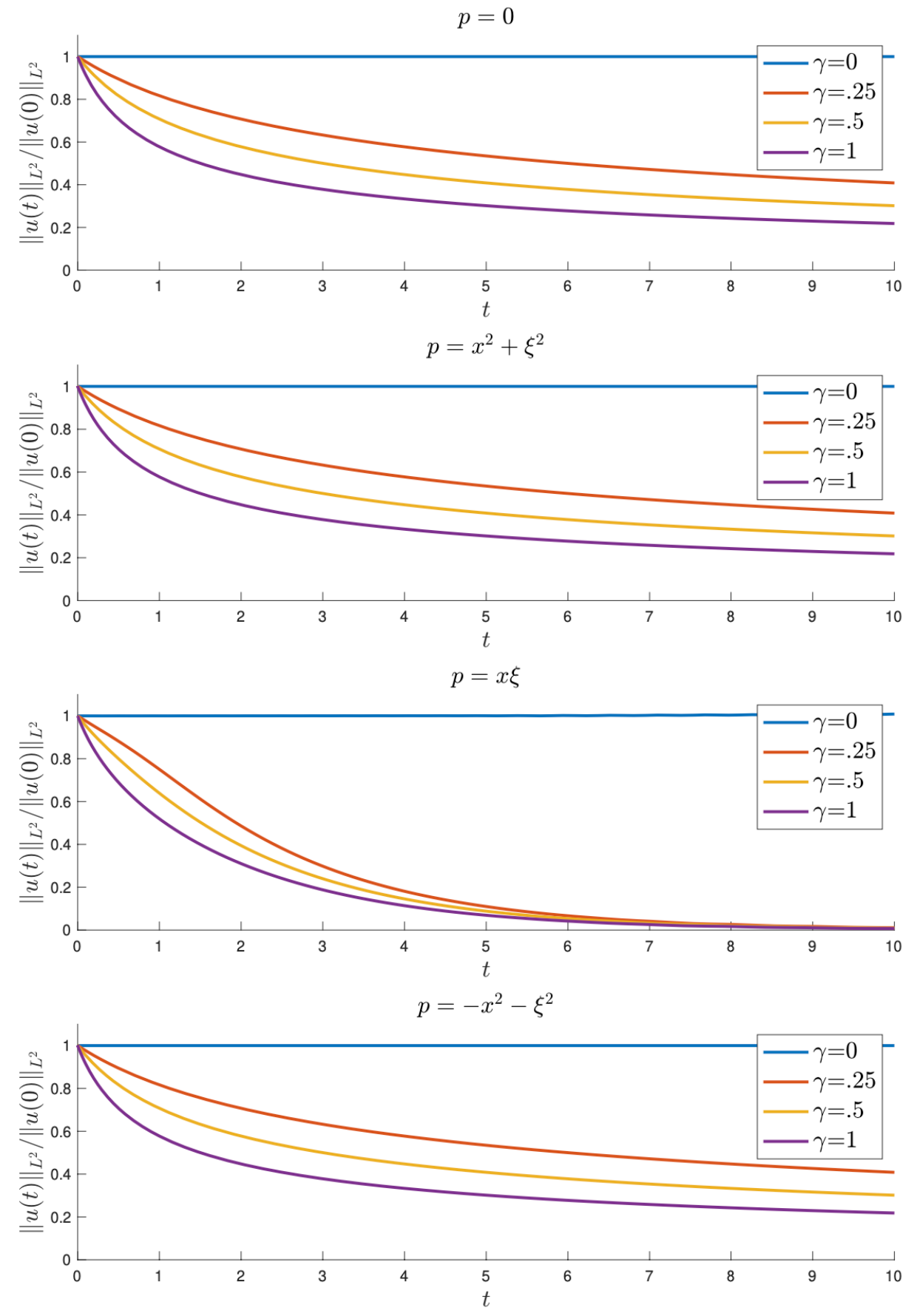}
\caption{The $\|u(t)\|_{L^2}/\|u(0)\|_{L^2}$ for the solution to~\eqref{e:quadraticFP} with $A=2I$ for various choices of $B$ in~\eqref{e:quadraticP}.}
\end{figure}

\subsection{Confining Hamiltonians}
We next consider the case where the Hamiltonian $p$ is confining. We assume in this subsection that there are $c,m>0$ and $M\in \mathbb{R}$ such that
\begin{equation}
\label{e:pConfines}
\begin{gathered}
\Delta p \geq 0, \qquad p\geq c|\nabla p|^2,\qquad \text{ on }|p|\geq M, \\
|p|\geq c\langle (x,\xi)\rangle^m-{1}/{c},\qquad (x,\xi)\in \mathbb{R}^{2n}.
\end{gathered}
\end{equation}
We show in Proposition~\ref{p:confined} that under this assumption, for sufficiently dispersed initial data, the Hilbert--Schmidt norm of the Lindblad evolution is bounded from below for long times.

We will use a maximum principle type argument to show that, in the presence of a confining Hamiltonian, the Fokker--Planck evolution remains well confined in $L^1$ for long times. We start by constructing an effective barrier with which to apply the maximum principle.
\begin{lemm}
\label{l:barrier}
Suppose that~\eqref{e:pConfines} holds. Then for any $f\in C^\infty(\mathbb{R})$, such that $\supp f'\subset (M,\infty)$, $f,f''\geq 0$, $f''\in C_c^\infty$, there is $C>0$ such that, defining
\begin{equation}
\label{eq:defgt} 
g(t):=\frac{g(0)}{C\gamma h t g(0)+1},\qquad v(t,x,\xi):= \exp (-g(t)f(p(x,\xi))), 
\end{equation}
we have 
\begin{equation}
\label{eq:supers}
(\partial_t -H_p -\tfrac12\gamma h \Delta)v\geq 0, \ \ t \geq 0 .
\end{equation}
\end{lemm}
\begin{proof}
We calculate
\begin{equation}
\label{eq:calcsup}
\begin{split}
v_t-\gamma h\Delta v&= [-g'(t)f(p)+\gamma h ( (-g^2[f'(p)]^2+gf''(p))|\nabla p|^2 +g  f'(p)\Delta p )]v\\
&\geq [-g'(t)f(p)+\gamma h (-g^2[f'(p)]^2|\nabla p|^2)]v\\
&\geq [-g'(t)f(p)+C\gamma h (-g^2[f'(p)]^2p)]v.
\end{split}
\end{equation}
Since $f\geq 0$ and $|f''(p)|\leq C$, 
$$
|f'(p)|\leq C\sqrt{f(p)}.
$$
Hence, {using that $f''\in C_c^\infty$, we have that there is $p_0$ such that  $f'(p)=L$ for $p\geq p_0$ large enough and  $ f'(p) = 0 $ for $ p < M $, }
$$
f(p)\geq {\max(c[f'(p)]^2, L(p-p_0)+f(p_0))}\geq c [f'(p)]^2p.
$$
Thus, for $C>0$ large enough, and $ g $ given in \eqref{eq:defgt} (so that 
$
 -g'(t)-C\gamma h g^2 =  0
$)
the last inequality in \eqref{eq:calcsup} gives \eqref{eq:supers}.
\end{proof}

In the next lemma, we show that, given some apriori assumptions on the Fokker--Planck solution, we are able to confine the majority of its $L^1$ mass to a bounded set. As a consequence, we obtain that the $L^2$ norm cannot decay for long times.
\begin{lemm}
\label{l:maximum}
Suppose that~\eqref{e:pConfines} holds.  
Then, $ \forall \, R_0, c_1 > 0  \, \exists\, R_1 >0 \, \forall \, \delta, N> 0 \, \exists\, C_{N, \delta} $ such that if
\begin{equation}
\label{e:upperBound}
\begin{gathered}
0\leq a_0(x,\xi)\leq e^{-c_1\langle (x,\xi)\rangle^2/h},\qquad |(x,\xi)|\geq R_0,\ \ \ \ 0\leq a_0\leq 1
\end{gathered}
\end{equation}
then for any solution $a(t) \in  L^\infty(\{t>0\}\times \mathbb{R}^{2n}_{(x,\xi)}) $ to 
$$
(\partial_t -H_p-\tfrac12\gamma h \Delta)a=0,\qquad a(0)=a_0,
$$
we have
\begin{equation}
\label{e:uIsConfined}
\|a(t)\|_{L^1(\mathbb{R}^d\setminus B(0,R_1))}=C_{N,\delta}h^N\qquad 0\leq  t\leq h^{-1+\delta}\gamma^{-1}.
\end{equation}
In particular, 
\begin{equation}
\label{e:uL2}
\|a_0\|_{L^1}-C_{N,\delta}h^N\leq C\|a(t)\|_{L^2},\qquad 
\end{equation}
\end{lemm}
\begin{proof}
Let  $M_0\geq 0$ such that
\[ p\geq M_0 \ \ \Longrightarrow \  \ {c}\langle (x,\xi)\rangle^m/2 \leq p \leq C_1 \langle (x,\xi)\rangle^2 
\  \text{ and } \ |(x,\xi)|\geq R_0. \]
  Then, let $\psi\in C_c^\infty((M_0,M_0+1);[0,\infty))$ with $\int \psi =1$ and define $f(x):=\int_0^x\int_0^s\psi(t)dtds$ so that $f''(x)=\psi(x)$,  $\supp f\subset (M_0,\infty)$ with $f,f''\geq 0$. Let $R_1\geq 0$ such that $p\geq M_0+1$ on   $|(x,\xi)|\geq R_1$. 

Since $f(p(x,\xi )) = 0 $ on $|(x,\xi)|\leq R_0$, $f(p)\leq Cp\leq C\langle(x,\xi)\rangle^2$, and $a_0$ satisfies~\eqref{e:upperBound} there exists $c_0>0$ such that 
$$
\exp(-c_0f(p(x,\xi))/h)\geq a_0.
$$
We now apply Lemma~\ref{l:barrier} with $ g(0) = c_0/h $: for $ v $ in \eqref{eq:defgt} 
\[  ( \partial_t -H_p - \tfrac 12\gamma h\Delta ) ( v - a ) \geq 0 ,  \ \ \ \  t \geq  0. \]
The maximum principle~\cite[Theorem 1]{cos} then shows that $ 0\leq a\leq v$ and consequently, {using that $f(p)\geq cp\geq c\langle (x,\xi)\rangle^m$ on $\mathbb{R}^{2n}\setminus B(0,R_1)$,}
\[ 
\begin{split}
\|a(t)\|_{L^1(\mathbb{R}^{2n}\setminus B(0,R_1))} & \leq \|v(t)\|_{L^1(\mathbb{R}^{2n}\setminus B(0,R_1)}
\\ & \leq \int_{\mathbb{R}^{2n}\setminus B(0,R_1)} e^{  - c_0 \langle (x,\xi)\rangle^m  / (C\gamma t + h ) } dxd\xi \leq C e^{ -c_0 /(C\gamma t h + h)  } , 
\end{split} \]
from which~\eqref{e:uIsConfined} follows.

To obtain~\eqref{e:uL2}, observe that 
$$
\partial_t \int a dxd\xi = \int (H_p+\tfrac 12\gamma h\Delta )a dxd\xi =0.
$$
Hence, 
\begin{align*}
\|a_0\|_{L^1}=\|a(t)\|_{L^1}&\leq \|a(t)\|_{L^1(B(0,R_1))}+C_{N,\delta}h^N\\
&\leq CR_1^{n}\|a(t)\|_{L^2(B(0,R_1))}+C_{N,\delta}h^N\\
&\leq CR_1^n\|a(t)\|_{L^2}+C_{N,\delta}h^N.
\end{align*}
\end{proof}

Finally, we show that the Hilbert--Schmidt norm of the Lindblad evolution with a confining Hamiltonian can, in many cases, be effectively controlled from below.

\begin{prop}
\label{p:confined}
Suppose that $ \mathcal L $ is given by \eqref{e:Lindblad}, assumptions \eqref{eq:assa}, 
and~\eqref{e:pConfines} hold, that $\ell_j$'s are as in \eqref{eq:exam}.
If $  h^{2 \rho - 1 }\leq \gamma \leq h^{-1} $ for some $ 0 \leq \rho \leq \frac23 $
then 
\[ \exists\, C_0 >  0 \, \forall \, c_1 > 0, R_0 > 0, a_0 \in  S^{L^2}_\rho
\text{ with $a_0/\|a_0\|_{L^\infty}$ satisfying~\eqref{e:upperBound}, }N, \delta > 0 \, \exists \, C > 0 \]
such that 
 if $A(t)$ satisfies (in the notation of \S \ref{s:sym})
\begin{equation*}
\partial_t A ( t) = \mathcal L A ( t ) , \  \  \ A(0) = \Op (a_0 ) ,  
\end{equation*}
then, for $0\leq t\leq h^{-1+\delta}\gamma^{-1}$,
$$
(2\pi h)^{\frac{n}{2}}\tr A(0)-Ch^N\|a_0\|_{L^\infty}-C e^{C_0{h^2}\gamma t}
(1+\gamma)th^{(2- 3 \rho)}\leq C\|A(t)\|_{\mathscr{L}_2}.
$$
\end{prop}
\begin{proof}
By Theorem~\ref{t:2}, 
\begin{equation}
\label{e:approxError}
\|A(t)-\Op(a(t))\|_{\mathscr{L}_2}\leq C e^{C_0{h^2}\gamma t}
(1+\gamma)th^{(2- 3 \rho)}
\end{equation}
where $a(t)=U(t)a_0$. In particular, $a(t)$ satisfies
$$
(\partial_t -H_p- \tfrac12 \gamma h \Delta)a=0,\qquad a(0)=0,
$$
and $a(t)\in S_{\rho}^{L^2}$. Since $a(t)\in S_{\rho}^{L_2}$, {uniformly in $t>0$}, by the Sobolev embedding, {$a(t)\in L^\infty(\{t>0\}\times \mathbb{R}^{2n}_{(x,\xi)})$} and hence applying Lemma~\ref{l:maximum} then yields
\begin{equation}
\label{e:lowerTrace}
\begin{aligned}
(2\pi h)^{n}\tr A(0)-Ch^N\|a_0\|_{L^\infty}&=\|a_0\|_{L^1}-Ch^N\|a_0\|_{L^\infty}\\
&\leq C\|a(t)\|_{L^2}=C(2\pi h)^{\frac{n}{2}}\|\Op(a(t))\|_{\mathscr{L}_2}. 
\end{aligned}
\end{equation}
The Proposition now follows from combining~\eqref{e:approxError} and~\eqref{e:lowerTrace}.
\end{proof}

\appendix
\section{operators with quadratic symbol growth}

\label{a:app}

We start with the composition formula of operators quantizing symbols in $ S_{(k)} $ where
that space was defined in \eqref{eq:defSk}. 

\begin{prop}
\label{p:comp}
Suppose that $ a_j \in S_{(k_j)} $, $ j =1,2$. Then $ \Op (a_1 ) \Op (a_2 ) = \Op (b )$, where
for any $ N \geq \max(k_1, k_2 ) $, 
\begin{equation}
\label{eq:comp} b ( x, \xi , h )  - \sum_{ \ell=0}^{N-1} \frac1 {\ell!}  \left(\frac{h}{2i}
\sigma ( D_x, D_\xi , D_y, D_\eta ) \right)^\ell a_1 ( x, \xi) a_2 ( y, \eta ) |_{ x=y, \xi = \eta } 
\in h^N S_0 , 
\end{equation}
where $ \sigma $ is the standard symplectic form on $ \mathbb R^{2n} $. 
\end{prop}

\noindent
{\bf Remark.} Note that $ b $ in the statement of the proposition is not 
necessarily in an $ S_{(k)} $ class since they are not closed under multiplication.

\begin{proof}
We observe that $ S_{(k)} \subset S ( m_k )$, $ m_k ( x, \xi ) = ( 1 + |x| + |\xi | )^k$. Hence 
\cite[Theorem 4.18]{z12} applies and, writing $ z = ( x, \xi ) $, $ w = ( y , \eta ) $, 
\[ b ( z , h ) = \exp ( i h A ( D ) ) ( a_1 ( x) a_2 ( w ) )  |_{z = w}, \ \ \ 
A ( D_{z,w} ) = - \tfrac12 \sigma ( D_x, D_\xi , D_y, D_\eta ) . \]
By Taylor's formula, 
\[ b  ( z , h ) = \sum_{ \ell=0}^{N-1} \frac{1}{\ell!} ( i h A ( D ) )^\ell ( a_1 ( z ) a_2 ( w) )|_{
z = w} + R_N ( z , h ) \]
where 
\[ R_N ( z, h ) : \frac{1}{(N-1)!} ( 1 - t)^{N-1} e^{ i t h A ( D ) } ( i h A ( D ) )^{N } ( 
( a_1 ( z ) a_2 ( w) )|_{z, w} . \]
For $ N \geq \max(k_1, k_2 ) $, $  A( D)^N a_1 ( z ) a_2 ( w ) \in S_0 ( \mathbb R^{4n}_{z,w}  ) $ and since $ e^{ i h t A ( D ) } : S_0 ( \mathbb R^{4n} )  \to S (  \mathbb R^{4n}) $ (with uniform bounds
for $ 0 \leq t \leq 1 $ -- see \cite[Theorem 4.17]{z12}) we conclude that 
$ R_N \in h^N S_0 (\mathbb R^{2n} ) $ which is \eqref{eq:comp}.
\end{proof}

We now present a general spectral result following the proof of a special case in \cite{hor} (see the example in \cite[\S C.2.2]{z12}):

\begin{prop}
\label{p:spec}
Suppose that $ p ( x, \xi ) \in C^\infty ( \mathbb R^d \times \mathbb R^d ) $ satisfies
\begin{equation}
\label{eq:assp}
 | \partial^\alpha p ( x, \xi ) | \leq C_\alpha ( 1 + |x| + |\xi| ) , \ \ \ |\alpha | \geq 1, 
\end{equation}
and define 
\[  \begin{split}
& N_p u = p^{\rm{w}} ( x, D ) u ,  \ \ \mathcal D ( N_p ) := \mathscr S ( \mathbb R^d ) , \\ 
& M_p u = p^{\rm{w}} ( x, D ) u ,  \ \ \mathcal D ( M_p ) := \{ u \in L^2 ( \mathbb R^d ) : 
 p^{\rm{w}} ( x, D ) u \in L^2 ( \mathbb R^d ) \} , 
 \end{split} 
 \]
 where in the case of $ u \in L^2 ( \mathbb R^d ) \subset \mathscr S' ( \mathbb R^d ) $
 we consider $  p^{\rm{w}} ( x, D ) u  \in \mathscr S' ( \mathbb R^d ) $. Then 
 $ M_p $ is closed and 
 \begin{equation}
 \label{eq:closure}
 M_p = \overline{N_p } , \ \ \  M_p^* = M_{\bar p } . 
 \end{equation}
  \end{prop}

\begin{proof}
\medskip
\noindent
We recall that  $ p^{\rm{w}} ( x, D ) : {\mathscr S}' \to {\mathscr S}' $ is continuous
and hence, if $  u_j  \to u $ and
$ p^{\rm{w}} ( x, D ) u_j  \to v $ in  $ L^2 $, then $ u_j \to u $ in 
$ {\mathscr S}' $. Consequently, $ v = p^{\rm{w} } ( x, D ) u \in L^2 $, 
 $ u \in {\mathcal  D} ( M_p ) $ and $M_pu = v$. This shows that $ M_p $ is closed.

 To show that $ M_p $ is the closure of $ N_p $ we have to 
show that for any $  u \in {\mathcal D} ( M_p  )$ there exists a family $ u_\varepsilon \in
{\mathscr S} $ such that $ u_\varepsilon \to u $ and 
\begin{equation}
\label{eq:closure1}  p^{\rm{w}} ( x, D
) u_\varepsilon \to p^{\rm{w}} ( x , D ) u \ \ \  \text{ in $ L^2 $ as   $ \varepsilon \to
0 $.}
\end{equation}
To construct $ u_\varepsilon $ we take
$\chi \in \CIc ( \RR^{2d} ) $ equal to one in $ B_{\mathbb R^{2n}} ( 0, 1 ) $,  and put  
\[  u_\varepsilon := \chi_\varepsilon^{\rm{w}} ( x, D ) u \in {\mathscr S }, \ \ \ 
\chi_\varepsilon ( x, \xi ) := \chi ( \varepsilon x , \varepsilon \xi ), \ \ \ 
 \ u_\varepsilon \to u \ \text{ in $L^2$.}  \] 
Then 
$  p^{\rm{w}}  u_\varepsilon  = \chi_\varepsilon^{\rm{w}} 
p^{\rm{w}} u_\varepsilon+ [ p ^{\rm{w}}  , \chi_\varepsilon
^{\rm{w}} ] u_\varepsilon
 $ 
and as $ \chi_\varepsilon^{\rm{w}} 
p^{\rm{w}}  u \to p^{\rm{w}}  $ in $ L^2 $,  we need to show
that 
\begin{equation}
\label{eq:commp} [ p ^{\rm{w}} ( x, D ) , \chi_\varepsilon ^{\rm{w}} ( x, D ) ] u \to 0  \ \text{ in $ L^2 $ }  , \ \ 
\ \varepsilon \to 0.
  \end{equation}
To see this we note that \cite[Theorem 4.18]{z12} and the two term Taylor expansion of
$ e^{ i A ( D ) } $ give
\begin{equation}
\label{eq:decompp}
[ p^{\rm{w}} ( x, D ) , \chi_{\varepsilon } ( x, D ) ] 
= a_\varepsilon^{\rm{w}} ( x, D ) , 
\ \ \ a_\varepsilon ( x, \xi ) = i \{ \chi_\varepsilon , p \} ( x, \xi ) + b_\varepsilon ( x, \xi ) ,
\end{equation}
where 
\[ 
b_\varepsilon ( x, \xi ) := \int_0^1 ( 1 - t ) \left(e^{ it A ( D)  } 
( i A( D)  )^2 \left( p ( x, \xi ) \chi_\varepsilon ( y, \eta ) 
- p ( y, \eta ) \chi_\varepsilon ( x, \xi ) \right)\right) |_{ x = y, \xi= \eta } dt ,  
 \]
and  $ A ( D ) :=  \sigma ( D_x, D_\xi;  D_y , D_\eta ) $.
In view of \eqref{eq:assp}, 
\begin{equation*}
 \{ \chi_\varepsilon , p \} ( x, \xi ) = \varepsilon \sum_{j=1}^n 
\left(  \partial_{x_j} p ( x, \xi ) (\partial_{\xi_j} \chi) (\varepsilon x, \varepsilon \xi )  
-  \partial_{\xi_j} p ( x, \xi ) (\partial_{x_j} \chi) (\varepsilon x, \varepsilon \xi )  \right) 
 \end{equation*}
is bounded in $ S ( 1 ) $, uniformly as $ \varepsilon \to 0 $. 

To obtain estimates on $ b_\varepsilon $ we observe that, for some  $c_{\alpha\beta}\in\mathbb{C}$,
\begin{align*}
&(iA(D))^2 (p( x, \xi ) \chi_\varepsilon ( y, \eta ) 
- p ( y, \eta ) \chi_\varepsilon ( x, \xi ))\\
&=\sum_{|\alpha|=|\beta|=2}c_{\alpha\beta}\varepsilon^2( \partial^\alpha p(x,\xi)\partial^\beta \chi(\varepsilon y,\varepsilon \eta)-\partial^\alpha \chi(\varepsilon x,\varepsilon \xi)\partial^\beta p(y,\eta))\in \varepsilon ( S (m) + S (1/m )) , 
\end{align*}
where the order function is given by $ m ( x, \xi , y , \eta ) := \langle x, \xi \rangle \langle y , \eta \rangle^{-1} $. 
The inclusion follows from the fact  $\varepsilon \leq C \langle \varepsilon (x,\xi)\rangle^{-1}$ for $(x,\xi)\in \supp\chi$ and from  the assumption~\eqref{eq:assp}. By~\cite[Theorem 4.17]{z12}, the operators $e^{ithA(D)}:S(m^{\pm 1} )\to S(m^{\pm 1} ) $ are bounded uniformly in $ t $. Since
$ m|_{ x = y, \xi=\eta} = 1 $,  $b_\varepsilon \in \varepsilon S(1)$, and 
\cite[Theorem 4.23]{z12} gives, uniformly as $ \varepsilon \to 0$,
\begin{equation}
\label{eq:commu}
\| b_\varepsilon^{\rm{w}} ( x,  D)  \|_{ L^2 \to L^2 } \leq C \varepsilon  . 
\end{equation}

We now choose $ \psi \in \CIc ( \RR^n )
$ supported in $ B_{\mathbb R^{2n} } ( 0 , 1 ) $, equal to one near $ 0 $, 
and  put  $ \psi_\varepsilon ( x , \xi ) = \psi (\varepsilon x, \varepsilon \xi ) $. 
Then $   \{ \chi_\varepsilon , p
\} (x, \xi ) \psi_\varepsilon ( x, \xi ) \equiv 0 $, and \cite[Theorems 4.18 and 4.23]{z12} imply 
\begin{equation*}
\| \{ \chi_\varepsilon , p\}^{\rm{w}} ( x,  D)  \psi_\varepsilon^{\rm{w}} (x , D ) \|_{ L^2 \to L^2 } \leq C \varepsilon .  
\end{equation*} 
This and 
\eqref{eq:commu} give
\[ \begin{split}  [ p ^{\rm{w}} ( x, D ) , \chi_\varepsilon ^{\rm{w}} ( x, D ) ]
 &  =   \{ \chi_\varepsilon , p \}^{\rm{w}}( x, D )(1-\psi_\varepsilon^{\rm{w}} (x, D))+ 
 \{ \chi_\varepsilon , p\}^{\rm{w}} ( x,  D)  \psi_\varepsilon ^{\rm{w}} (x, D  ) 
   + b_\varepsilon^{\rm w} ( x, D ) 
 \\ &  = 
 \{ \chi_\varepsilon , p \}^{\rm{w}} ( x , D) ( 1 - \psi_\epsilon^{\rm{w}} (x, D) ) + 
 \mathcal O( \varepsilon )_{L^2 \to L^2 } .  
 \end{split} \]
Since $ \psi_\varepsilon^{\rm{w}} ( x, D)  u \to u $ in $ L^2 $ and $\{\chi_\varepsilon,p\}\in S(1)$
(hence by \cite[Theorem 4.23]{z12} $ \| \{ \chi_\varepsilon , p \}^{\rm{w}} ( x , D) \|_{L^2 \to L^2} $ is
uniformly bounded), this and \eqref{eq:commu} give
 \eqref{eq:commp}. 
 
It remains to show the last assertion in \eqref{eq:closure}. For that 
we recall that $ v \in {\mathcal D} ( M_p^* ) $ if and only if there exists  $ C = C(v) $
such that for all $ u \in {\mathcal D } ( M_p )  $
\begin{equation} 
\label{eq:DMp}  \langle M_p u , v \rangle \leq C \| u \|_{L^2} . 
\end{equation}
For $ u \in \mathscr S \subset {\mathcal D } ( M_p ) $ we have
$ \langle M_p u , v \rangle = \langle u , \bar p^{\rm{w}} ( x, D ) v
\rangle, $  
where $ {\bar p}^{\rm{w}} ( x, D ) v \in {\mathscr S } ' $ and 
\eqref{eq:DMp} implies that $ {\bar p} ^{\rm{w}} ( x, D ) v \in L^2 $. 
Hence $  M_p^*  \subset M_{\bar p }$.
Since $ M_p^* $ is closed, 
$ N_{\bar p} \subset N_p^* = {\bar N_p}^* = M_p^*.
$
It follows that 
$ M_{\bar p } = \bar N_{\bar p } \subset M_p^* $ and that $  M_p^*  =  M_{\bar p }$.
\end{proof}

\section{by Zhen Huang and Maciej Zworski}
\label{a:num}

\begin{figure}[h]
\centering
\includegraphics[width=16cm]{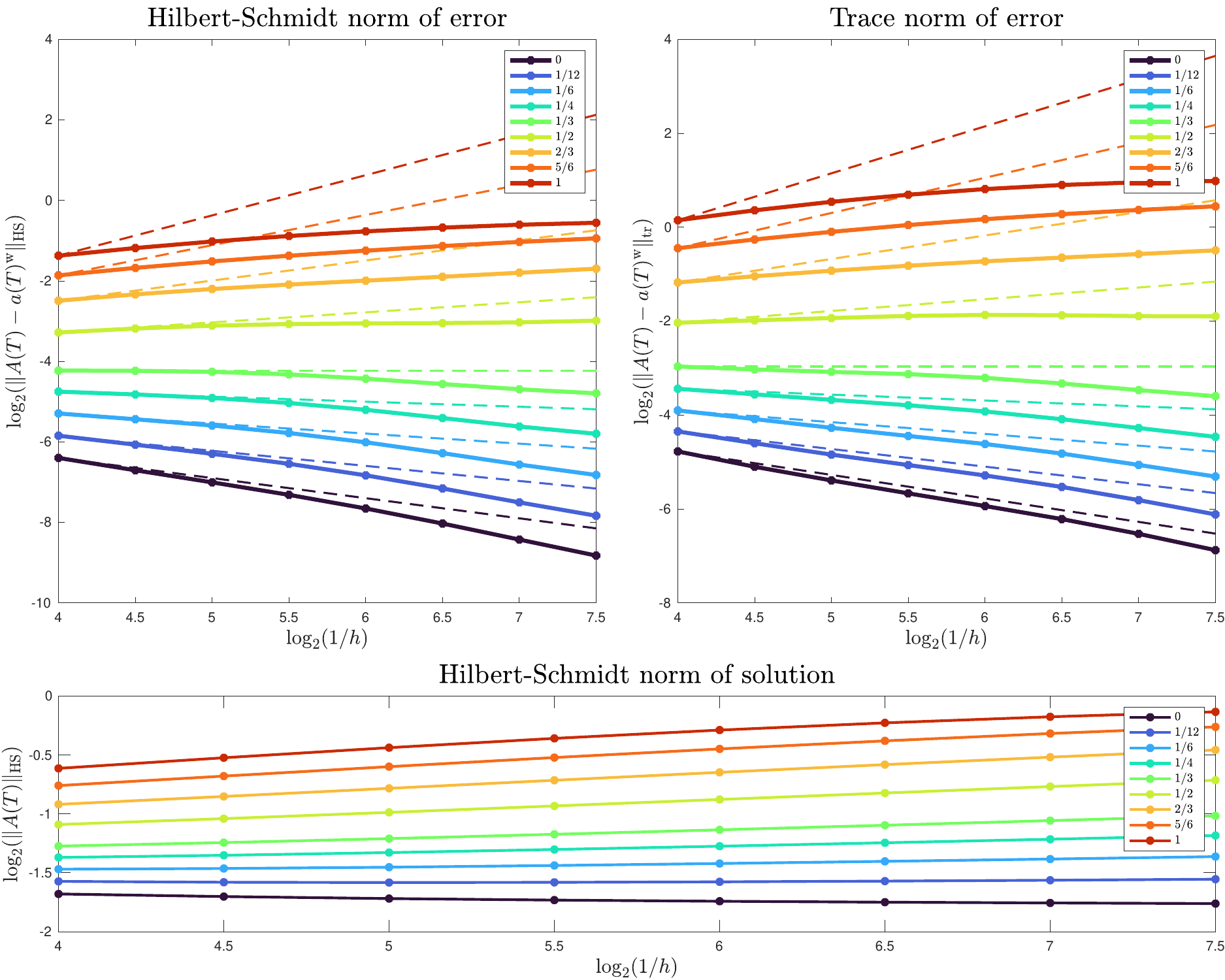}
\caption{Comparison of Lindblad and Fokker--Planck evolutions for $P$ as in~\eqref{eq:defV} and $L_j$ as in~\eqref{eq:defL} at time $ T = 2 $ and for initial data given by a coherent state \eqref{eq:a0} with $ h_0 = h $, $ x_0 = -1$, $ \xi_0 = 0 $, and 
$ \gamma = h^\delta $,  $ \delta = 0, \frac1{12} , \frac16, \frac14, \frac13, \frac12, \frac23, \frac56, 1$. The figure illustrates Theorem \ref{t:1Gauss} and the estimates \eqref{eq:HRR} from \cite{hrr}, with the dashed lines given by $ \log(1/h) \to \log h^{ \frac{3 \delta- 1 }2} $, which corresponds to the exponents in those estimates. 
The change of behaviour at $ \delta = \frac13 $ (the norms of the error decrease with $h$ for $\delta\leq \frac{1}{3}$) suggests that the bounds in \eqref{eq:nofric0} and
\eqref{eq:HRR} are accurate but we do not have conclusive evidence, see Figure~\ref{f:uh_Th}. 
We also show the behaviour of the Hilbert--Schmidt norm
which is consistent with the results of \S \ref{s:lindbladBound} (see the Remark after Theorem \ref{t:1Gauss}). \label{f:uh_T=2}}
\end{figure}

{We describe the results of numerical experiments illustrating the difference between the Schr\"odinger 
\eqref{eq:Schr} and Lindblad  \eqref{e:Lindblad} evolutions, and their relation to the corresponding 
classical evolutions given by the Hamiltonian flow and the Fokker--Planck equation \eqref{eq:FP}, 
respectively. In particular, for both the Hilbert--Schmidt and trace class norms, we compare the results to the bounds 
\eqref{eq:nofric0}, \eqref{eq:nofric}, and the corresponding bounds in \cite[Theorem 1.1]{hrr}.}

\begin{figure}[h]
  \centering
  \includegraphics[width=16cm]{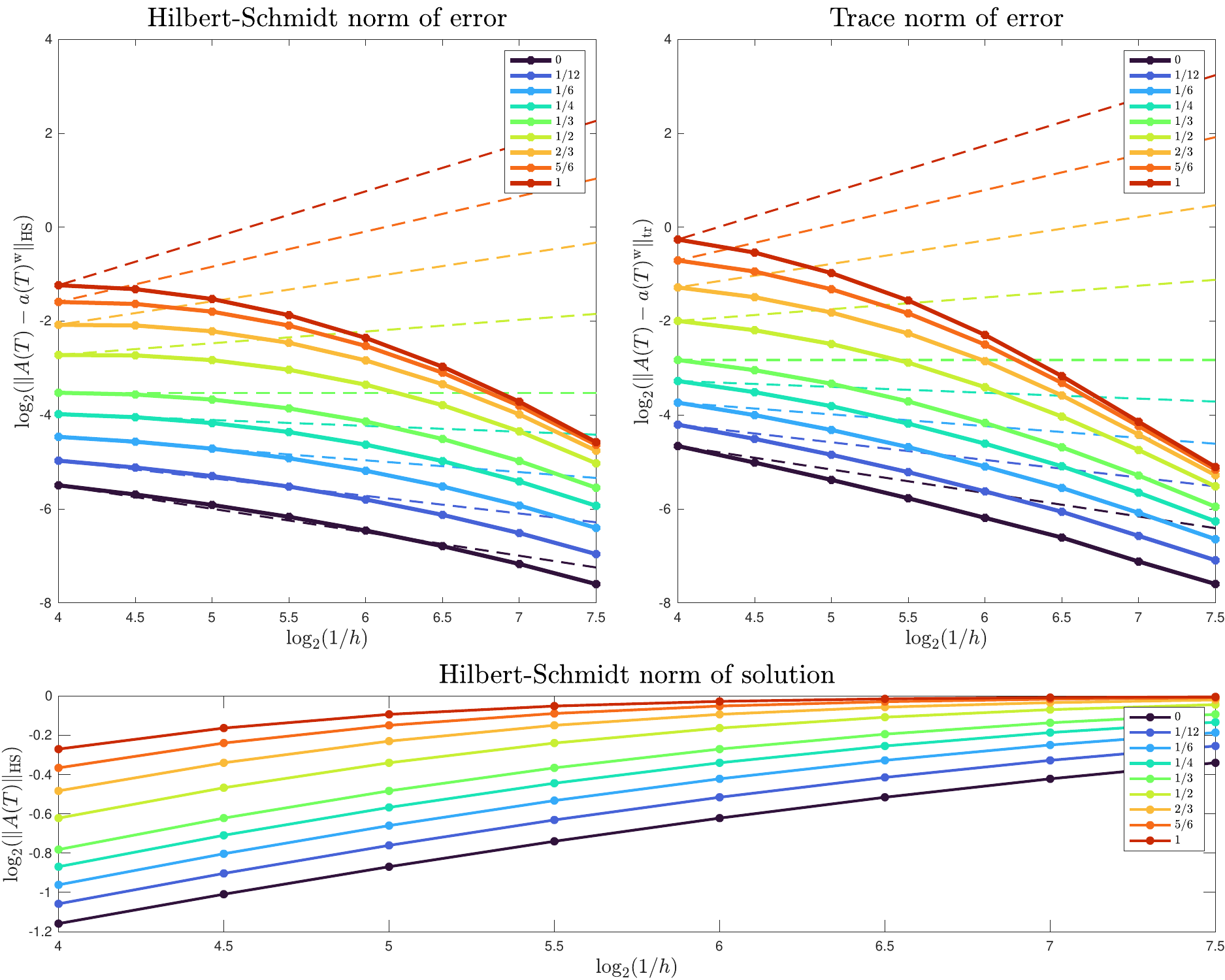}
  \caption{An analogue of Figure \ref{f:uh_T=2} for the initial data given by $ a_0 $ with 
  $ h_0  = 2^{-3} $, $ x_0=-0.8$, $ \xi_0 = 0 $. The agreement of the evolutions seems dramatically better than the estimates given in Theorem \ref{t:1}, especially in view of the fact that those estimates are only non-trivial for 
  $ \gamma = h^\delta $ and $\delta \leq \frac13$. 
  The behaviour of the Hilbert Schmidt norm is consistent with the results of \S \ref{s:lindbladBound}  (see Remark 2 after Theorem \ref{t:1}). 
  \label{f:ufix_T=2}}
  \end{figure}

{If the Hamiltonian acts on $ L^2 ( \mathbb R^n ) $ then the computations for the Lindblad evolution 
and the Fokker--Planck equation have to be performed in $ \mathbb R^{2n} $, which is a dramatic increase of dimension. This forces us, at this early stage, to restrict our attention to $ n = 1 $ which does not allow for chaotic behaviour. However, see \cite[\S 1.3]{hrr} for pointers to the physics literature where time 
dependent chaotic systems were considered.}

{We choose a one dimensional model in which a hyperbolic point occurs in the classical dynamics and hence 
we observe a $ \log 1/h $ Ehrenfest time: the classical/quantum correspondence breaks down 
at that time -- see \eqref{eq:eg} and Figures~\ref{f:1},\ref{f:2}. The model is given by the Schr\"odinger operator with 
a double well potential:
\begin{equation}
\label{eq:defV}
\begin{gathered}
  P=\left(h D_x\right)^2+V(x), \quad D_x=(1 / i) \partial_x,\\
    V(x)=\left(x^2-\tfrac14 \right)^2.  
\end{gathered} 
\end{equation}
We choose the jump operators, $L_j$, so that the Fokker--Planck equation \eqref{eq:FP} takes
the simple form \eqref{eq:exam}:
\begin{equation}
\label{eq:defL}
  L_1=x, \quad L_2=h D_x, \ \ \ \ \ L_j^* = L_j . 
\end{equation}
The specific choice of constants in $ V $ appearing in \eqref{eq:defV} is dictating by numerical considerations: we need our dynamics to be confined to a box of size $ [ 2 \pi , 2 \pi ] $.}

{We write the Lindblad evolution equation \eqref{e:Lindblad} as a differential equation for the
Schwartz kernel of  $ A $, $ A ( x, y ) $, $ x, y \in \mathbb R$:
\begin{equation}
\label{eq:LinA}
  \begin{gathered}
    \partial_t A(t)=\mathcal{L} A(t), \quad A(0)=A_0, \ \ \ A = A ( x, y ) , \\
      (\mathcal{L} A) ( x , y )  =\frac{i}{h}\left(-\left(h \partial_x\right)^2+V(x)+\left(h \partial_y\right)^2-V(y)\right) A-\frac{\gamma}{2 h}(x-y)^2 A .
  \end{gathered} 
  \end{equation}
 The Fokker-Planck equation \eqref{eq:FP} is an evolution equation for functions in phase space, that is, functions of
 $ ( x, \xi ) $:
  \begin{equation}
    \begin{gathered}
      \partial_t a(t)=Q a(t), \quad a(0)=a_0, \ \ \    a ( t ) = a ( t, x, \xi ) , \\
      Q=2 \xi \partial_x-V^{\prime}(x) \partial_{\xi}+\tfrac{1}{2} \gamma h\left(\partial_x^2+\partial_{\xi}^2\right) .
    \end{gathered}
    \label{eq:FPA}
  \end{equation}}

\begin{figure}[h]
  \centering
  \includegraphics[width=160mm]{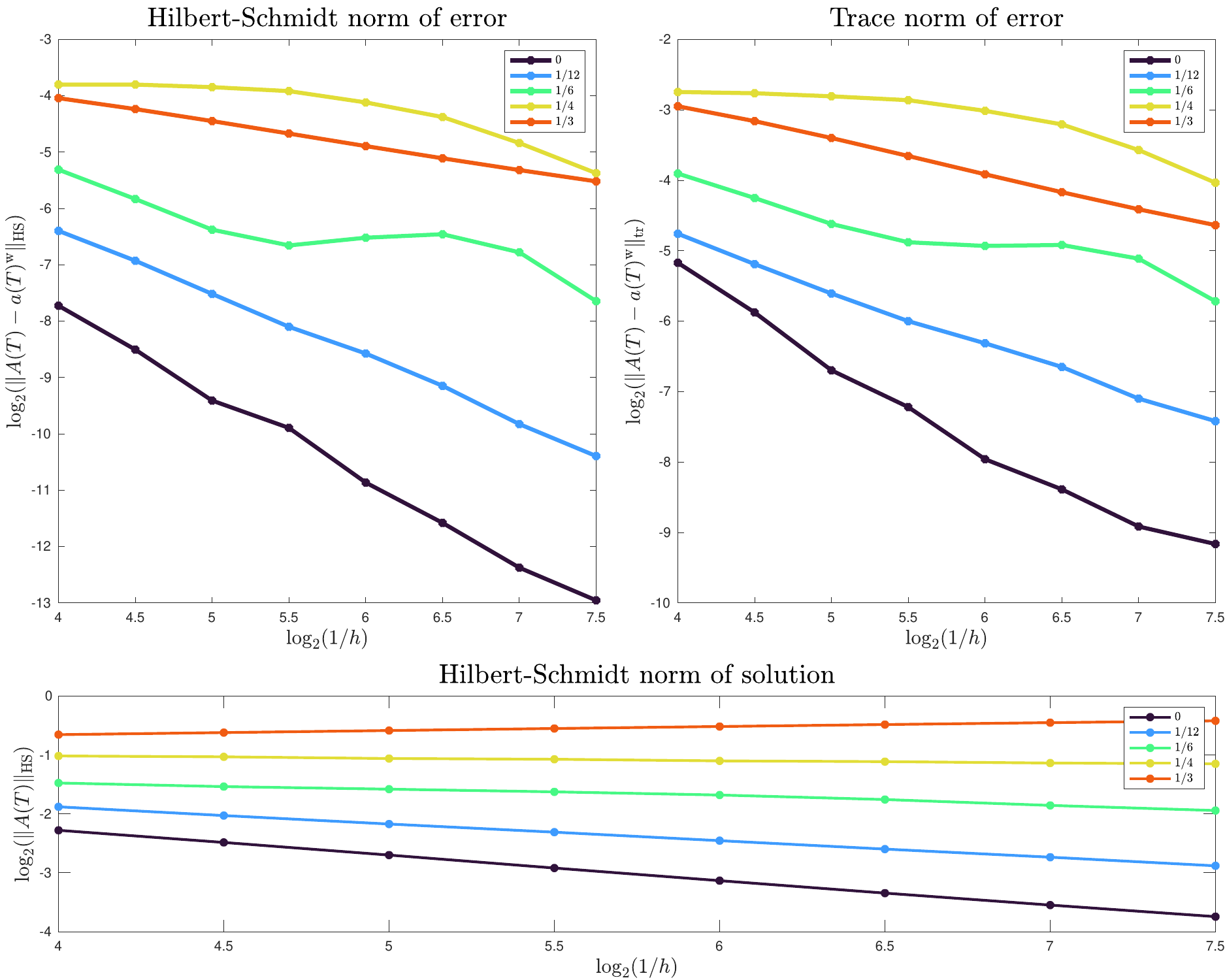}
  \caption{For the initial data in Figure~\ref{f:uh_T=2} (a coherent state with an $ h$ dependent 
  symbol) we take $ \gamma =  h^\delta $ with $ \delta= \frac13, \frac14, \frac16, \frac1{12}, 0$ and time equal to $T=h^{-\frac{1-3\delta}{2}}$ which is the limit of the validity of estimates
  in Theorem \ref{t:1Gauss} and in \eqref{eq:HRR} (from \cite{hrr}). 
   The results confirm their validity but do not indicate optimality in our model.\label{f:uh_Th}}
    \end{figure}
  
  \begin{figure}[h]
    \centering
    \includegraphics[width=160mm]{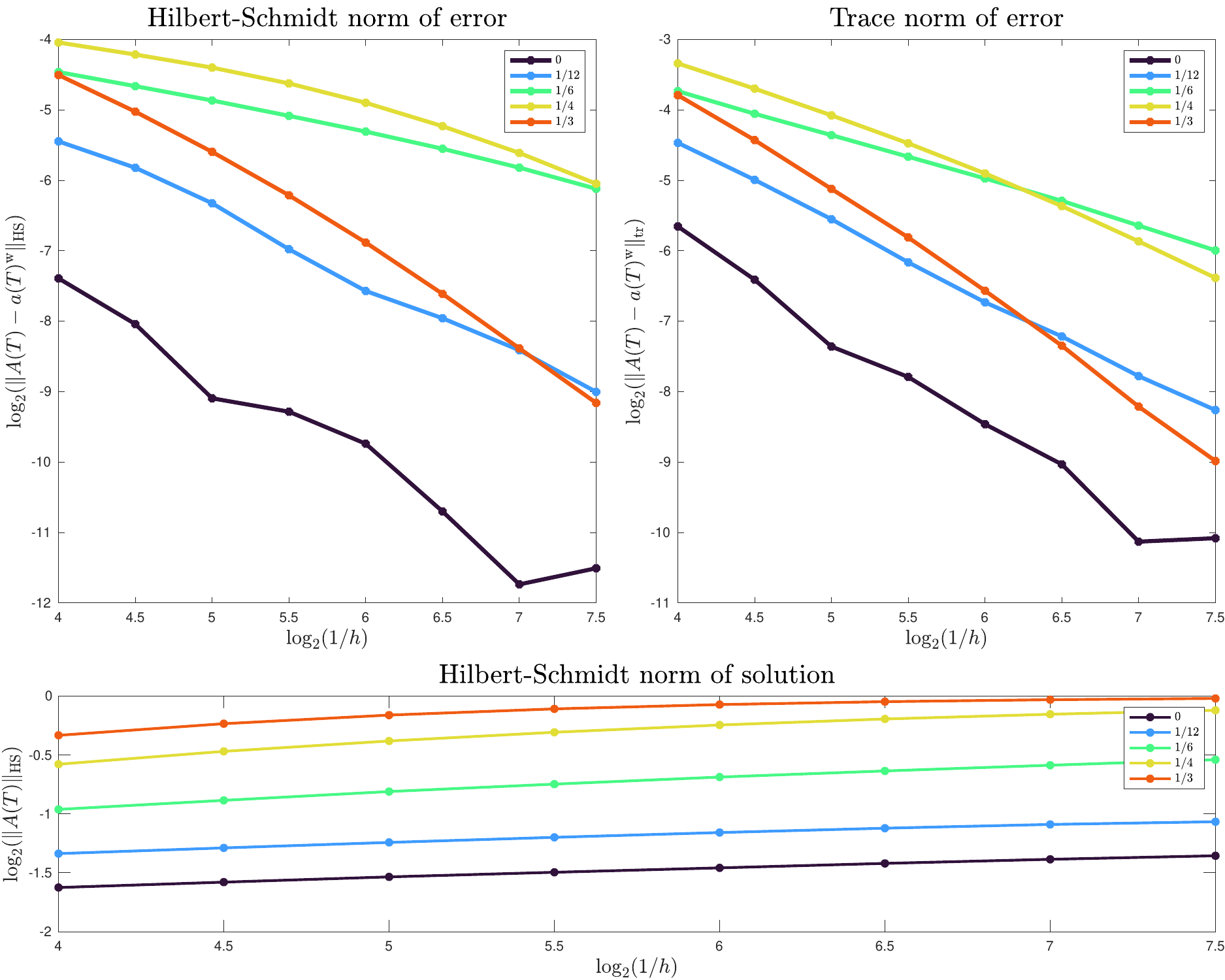}
    \caption{The analogue of Figure \ref{f:uh_Th} with the initial data from Figure~\ref{f:ufix_T=2}.
    It confirms the validity of the estimates in Theorem \ref{t:1} but suggests that in the model we considered they might not be optimal.
        \label{f:ufix_Th}}
    \end{figure}

\subsection{Numerical schemes} \label{s:num}
{We use the following numerical methods to solve \eqref{eq:FPA}. The operator 
 $Q$ is split into two parts, $Q=Q_1+Q_2$, where
\begin{equation}
\label{eq:Q2Q12}
    Q_1:=2 \xi \partial_x + \tfrac{1}{2}\gamma h \partial_x^2,\quad Q_2:=-V^{\prime}(x) \partial_{\xi}+\tfrac{1}{2} \gamma h\partial_{\xi}^2.
\end{equation}
On the Fourier transform side in the $ x $ variable (with the Fourier transform variable denoted by $ k_x $), the operator $Q_1$ acts  by multiplication:
$$
\widehat Q_1 \widehat a_1 ( k_x, \xi ) := \widehat {Q_1 a_1} (k_x, \xi) = \left(2 i\xi k_x - \tfrac12 {\gamma h} k_x^2\right) \widehat a_1(k_x, \xi).  
$$
and thus the corresponding evolution operator could also be written explicitly:
$$
 e^{t \widehat Q_1}\widehat a_1(k_x, \xi) =  e^{t\left(2i\xi k_x - \frac12{\gamma h} k_x^2\right)}\widehat a_1(k_x, \xi).
$$
Similarly, on the Fourier transform side in the $\xi$ variable (with the Fourier variable denoted by 
$ k_\xi$), the operator $Q_2$ acts diagonally:
$$
\widehat Q_2\widehat a_2(x, k_{\xi}) :=\widehat { Q_2 a_2} (x, k_{\xi}) = \left(- i V^{\prime}(x)k_{\xi} - \tfrac12 {\gamma h}k_{\xi}^2\right) \widehat a_2(x, k_{\xi}).
$$
and the corresponding evolution operator is
$$
 e^{t \widehat Q_2}\widehat a_2(x, k_{\xi}) =  e^{t\left(-iV^{\prime}(x)k_{\xi} - \frac12{\gamma h}k_{\xi}^2\right)}\widehat a_2(x, k_{\xi}).
$$
This motivates the following operator-splitting scheme for time-stepping, known as the (second-order) Lie-Suzuki-Trotter formula: 
$$
\begin{aligned}
   e^{n\Delta tQ} = \left( e^{\Delta tQ} \right)^n= &\left( e^{\frac{\Delta t}{2}Q_1} e^{\Delta tQ_2} e^{\frac{\Delta t}{2}Q_1}\right)^n+\mathcal O((\Delta t)^2)\\
 = &  e^{\frac{\Delta t}{2}Q_1}\left( e^{\Delta tQ_2} e^{\Delta tQ_1}\right)^{n-1} e^{\Delta tQ_2} e^{\frac{\Delta t}{2}Q_1}+\mathcal O((\Delta t)^2). 
\end{aligned}
$$
Numerically, the Fourier transform in 1D is handled by the 1D FFT function.}

To solve the Lindblad equation \eqref{eq:LinA}, we first change variable and define
$$
B(w, z):=A(w+z, w-z).
$$
The Lindbladian then becomes
\begin{equation}
\label{eq:L2B}
  \mathcal{L}=\frac{i}{h}\left(-h^2 \partial_z \partial_w+V(w+z)-V(w-z)\right)-\frac{\gamma}{2 h}\left(4 z^2-h^2 \partial_w^2\right). 
\end{equation}
As in \eqref{eq:Q2Q12} we write $\mathcal L = \mathcal L_1 + \mathcal L_2$,
\begin{equation*}
    \mathcal L_1:=\frac{i}{h}\left(V(w+z)-V(w-z)\right)-\frac{2\gamma}{ h}z^2,\quad
    \mathcal L_2 : =-ih \partial_z \partial_w+\frac{\gamma h}{2 }\partial_w^2.
\end{equation*}
The operator $\mathcal L_2$ acts by multiplication on the Fourier transform side:
$$
\mathcal L_2 \widehat B (k_w, k_z) = \left(i h k_w k_z - \frac{\gamma h}{2}k_w^2\right) \widehat B(k_w, k_z),
$$
while $\mathcal L_1$ is already a multiplication operator in the physical space.  We again proceed using the 
Lie--Suzuki--Trotter formula. {The MATLAB codes are included in \S \ref{s:cod}.}


\subsection{Numerical experiments}
\label{s:exp}
The numerical results are shown in figures in \S \ref{s:intr} and in this Appendix and we refer
to their captions for the interpretation of the results. 

When comparing the quantum (Lindblad) and classical (Fokker--Planck) evolution
in Hilbert--Schmidt norms we take
\[   A ( 0 ) = a( x, h D, h  ) , \ \ \ a( x, \xi , h ) = \sqrt{ 2 \pi h } \, a_0 ( x, \xi ), \]
where 
\begin{equation}
\label{eq:a0}
  a_0  (x,\xi)= \sqrt{\frac{2}{\pi h_0}} e^{-\frac{(x-x_0)^2}{h_0}-\frac{(\xi-\xi_0)^2}{h_0}}, \ \ \
  \| a_0 \|_{L^2} = 1 . 
\end{equation}
The case of $ h_0 $ fixed but small (so that $a_0 $ is numerically compactly supported) 
corresponds to the assumptions in Theorem \ref{t:1} and $ h_0 = h $ to those in 
Theorem \ref{t:1Gauss}, and in \cite{hrr}. When we compare trace class norms we normalize
the initial data so that $ \| A ( 0) \|_{\mathscr L_1} = 1 $. 

The initial value for Lindblad equation \eqref{eq:L2B} (for the rotated Schwartz kernel $ B $) is 
\begin{equation}
B(w,z) = \frac{2}{\sqrt{\pi h}} e^{-\frac{(w-x_0)^2}{h_0}} e^{-\frac{h_0}{h^2}z^2}  e^{i\frac{2\xi_0 z }{h}}.
\end{equation}
We recall the general relation between 
$a(x,\xi)$ and $B(w,z)$: 
\begin{equation}
  \begin{aligned}
    B(x,y) = \sqrt{\frac{2}{\pi h}}\int a(x,\xi) e^{i\frac{2y\xi }{h}} d\xi, \ \ \ \ 
    a(x,\xi) = \sqrt{\frac{1}{2\pi h}}\int B(x,y) e^{-i\frac{2\xi y}{h}} dy.
  \end{aligned}
  \label{eq:relation}
 \end{equation}

\subsection{Matlab codes}
\label{s:cod}
The first code solves the Fokker--Planck equation as described in \S \ref{s:num}:
\begin{Verbatim}[numbers=left]
function u = fp(dV, h, gamma_x, gamma_p, u, Nt, dt, N)
% solving Fokker-Planck equation on [-a,a] * [-a,a] using N * N grid points, 
% with initial data u, time step dt, for Nt time steps,
% with force dV(x), semiclassical parameter h,
% and diffusion coefficients gamma_x and gamma_p.

% Setup grids
a = pi; dx = 2*a/N;dp = 2*a/N;
x = -a:dx:a-dx; p = -a:dp:a-dp; [XX, PP] = ndgrid(x, p);
K = [0:N/2-1 0 -N/2+1:-1]*(pi/a);
[Kx, Kp] = ndgrid(K, K);

% construct exp(Q1*dt/2) and exp(Q2*dt)
Q1 = 2*PP.*(1j*Kx) - gamma_x*h/2.0 * (Kx.^2);
Q2 = -dV(XX).*(1j*Kp) - gamma_p*h/2.0 * (Kp.^2);
expQ1_2 = exp(Q1*dt/2); expQ1 = exp(Q1*dt); expQ2 = exp(Q2*dt);
for n = 0 : Nt-1
    % evolve Q1 for half a time step if t = 0
    if n == 0
        u1 = fft(u); 
        u1 = expQ1_2.*u1;
        u = ifft(u1);
    end

    %  evolve Q2 for a full time step
    u1 = transpose(fft(transpose(u)));
    u1 = expQ2.*u1;
    u = transpose(ifft(transpose(u1)));

    % evolve Q1 for a full time step
    u1 = fft(u); 
    if n < Nt-1
        u1 = expQ1.*u1;
    else 
        u1 = expQ1_2.*u1;  % evolve Q1 for a half time step if t = T
    end
    u = ifft(u1);
end
\end{Verbatim}
 
The next code gives the solution to the Lindblad equation \eqref{eq:LinA}: in the notation 
of \eqref{eq:L2B} we obtain $ B $ (denoted by {{rho}} in the code) with the initial data 
given by {{rho}}. An additional feature is allowing for different $ \gamma$'s in front of
$ \Delta_x $ and $ \Delta_\xi $ in \eqref{eq:exam}, $ \gamma_x $ and $ \gamma_p$ respectively.
In our experiment we take $ \gamma_x = \gamma_p $. 

\begin{Verbatim}[numbers=left]
  function rho = lindblad(V, h, gamma_x, gamma_p, rho, Nt, dt, N)
  % solving Lindblad equation on [-a,a] * [-a,a] using N * N grid points, 
  % with initial data u, time step dt, for Nt time steps,
  % with force V(x), semiclassical parameter h,
  % and diffusion coefficients gamma_x and gamma_p.

  % Setup grids
  a = pi; dw = 2*a/N;dz = 2*a/N;
  w = -a:dw:a-dw; z = -a:dz:a-dz; [ww, zz] = ndgrid(w, z);
  K = [0:N/2-1 0 -N/2+1:-1]*(pi/a);
  [Kw, Kz] = ndgrid(K, K);
  
  % construct exp(Q1*dt/2) and exp(Q2*dt)
  L1 = 1i*(V(ww+zz)-V(ww-zz))/h - 2*gamma_p*(zz.^2)/h;
  L2 = 1i*(Kw.*Kz)*h-h*gamma_x*(Kw.^2)/2;  
  expL1dt_2 = exp(L1*dt/2.0); expL1dt = exp(L1*dt); expL2dt = exp(L2*dt);
  
  for n = 0 : Nt-1
      % evolve L1 for half a time step if t = 0
      if n == 0
          rho_tmp = expL1dt_2.*rho;
      end
  
      %  evolve L2 for a full time step
      rho_hat = ifft2(rho_tmp);
      rho_hat = expL2dt.*rho_hat;
      rho_tmp = fft2(rho_hat);
  
      % evolve Q1 for a full time step
      if n < Nt-1
          rho_tmp = expL1dt.*rho_tmp;
      else 
          % evolve Q1 for a half time step if t = T
          rho = expL1dt_2.*rho_tmp;  
      end
  end
  \end{Verbatim}
 
The next matlab codes implements the transformations between $a(x,\xi)$ and $B(w,z)$ given in 
\eqref{eq:relation} via equidistance quadrature:
\begin{Verbatim}[numbers=left]
  function rho = a2B(h, p, u)
  % transform a(x,p) to B(w,z), h is the semiclassical parameter,
  % p is the grid points as defined before, i.e. p = -a:dp:a-dp;
  % u is function a(x,p) evaluated on the N*N mesh,
  % as defined in previous codes.

  Kernel_ift = exp(1j*2*p'*p/h)*dp/(sqrt(2));
  rho = u* Kernel_ift;

  end 

  function u = B2a(h, z, rho)
  % transform B(w,z) to a(x,p), h is the semiclassical parameter,
  % z is the grid points as defined before, i.e. z = -a:dz:a-dz;
  % rho is function B(w,z) evaluated on the N*N mesh,
  % as defined in previous codes.
  
  Kernel_ft = exp(-1j*2*y'*y/h)*dy*sqrt(2)/(pi*h);
  u = rho* Kernel_ft;
 
  end
\end{Verbatim}
Finally we give the code for evaluating the Hilbert-Schmidt norm and the trace norm of data given as $B(w,z)$.
Since $B(w,z)=A(w+z, w-z)$, therefore $\|A\|_{\text{HS}} = \frac{1}{\sqrt 2}\|B\|_{\text{HS}}$. To calculate the trace norm,
  we first reconstruct $A(x,y)$ on a $2N\times 2N$ grid from $B(w,z)$ on a $N\times N$ grid by rotation and interpolation,
  and then calculate the trace norm (also known as the nuclear norm in linear algebra) of this $2N\times 2N$ matrix using singular value decomposition (SVD). 
\begin{Verbatim}[numbers=left]
  function [HS_norm, Tr_norm] = norm_calc(h, rho, a, N)
  % calculate the Hilbert-Schmidt norm and the trace norm of B(w,z)
  % h is the semiclassical parameter, rho is the function B(w,z) 
  % evaluated on the N*N mesh,
  % a is the computation region, 
  % i.e. the calculation is done on [-a,a]*[-a,a],
  % N is the number of grid points in each direction.

  err_HS = sqrt((sum(abs(rho(:)).^2))*((2*a/N)^2) /2);

  rho_rot = rho_rotate(rho);
  err_trace = (sum(svd(rho_rot)))*(2*a/(2*N));
  end 

  function rho_rotate = rho_rotate(rho)
  % by rotation and interpolation, reconstruct A(x,y) from B(w,z):
  N = size(rho,1); rho_reflect = rho(:, N:-1:1);
  rho_rotate = zeros(2*N);
  for i = 1:2*N-1
      rho_diagonal = diag(rho_reflect, N-i);
      L = length(rho_diagonal);
      rho_fit = rho_diagonal;
      if L>1
          rho_fit = spline(1:L,rho_diagonal,1:0.5:L) ;
      end
      idy = (1:(2*L-1))-L + N+1;
      rho_rotate(i,idy) = rho_fit;
  end
  end
\end{Verbatim}

\end{document}